\documentclass[11pt,a4paper]{article}    %%%%%%%%Bug tetex 3.0 !!!!!
\usepackage{amsmath,amsfonts,amssymb,amsthm,amscd}
\usepackage[english]{babel}

\newtheorem{theorem}{Theorem}
\newtheorem{corollary}{Corollary}
\newtheorem{proposition}{Proposition}
\newtheorem{lemma}{Lemma}
\newtheorem{remark}{Remark}
\newcommand{\dsp}{\displaystyle}

\newcommand{\R}{\mathbb{R}}

\numberwithin{equation}{section}

\begin{document}
\title{Efficacy of a magnetic shield against a Vlasov-Poisson plasma}
\author{Silvia Caprino*,  Guido Cavallaro$^{+}$,   Carlo Marchioro$^{++}$}
\maketitle
\begin{abstract}
The aim of the paper is to test
on a simple model how impenetrable may be a magnetic shield.
We study the time evolution of a single species positive plasma, confined in the half-space $x_1>0$. 
The confinement is the result of a balance between a magnetic field and an external
field, both singular at $x_1=0$; the magnetic field forbids the entrance of plasma particles in the region
$x_1 \leq 0$, whereas the external field attracts  the plasma particles towards $x_1=0$.
The plasma has finite total charge and velocities distributed with a Maxwell-Boltzmann law. 
\end{abstract}


\footnotetext{*Dipartimento di Matematica Universit\`a Tor Vergata, via della Ricerca Scientifica, 00133 Roma (Italy), 
caprino@mat.uniroma2.it}
\footnotetext{$^+$Dipartimento di Matematica Universit\`a  La Sapienza, p.le A. Moro 2,  00185 Roma (Italy),  
cavallar@mat.uniroma1.it}
\footnotetext{$^{++}$International  Research Center M$\&$MOCS (Mathematics and Mechanics of Complex Systems),  
marchior@mat.uniroma1.it}

\section{Introduction and main result}

It is well known that a magnetic field that becomes very strong (or at the limit diverges) on a surface can forbid the entrance of a point charged particle. In recent papers  \cite{KRM12, SIAM, REND, TORO, JSP17, GH, K, ST1, ST2} it has been proved that a similar effect can happen for a Vlasov-Poisson plasma, for which  some part of the charged plasma could be pushed a priori against the surface by the other particles of the plasma. That is, the magnetic shield effect can be stronger than the internal interaction of the plasma. In the present paper we show that the same happens when the surface attracts the plasma, also when this attraction is very strong.
The analysis of the magnetic shield effect depends on the geometry of the surface and the total charge of the plasma. 
We discuss a simple case, but its extension to more complicated cases may be possible.

\noindent We consider how to defend  the half-space $\{ x\in {\mathbb{R}}^3: x_1<0\}$
from a Vlasov-Poisson plasma put in $x_1 >0$, by means of a magnetic shield.  The particles of the plasma are attracted towards the surface
$x_1=0$ by a force whose potential is 
\begin{equation}
U(x)=-\frac{1}{x_1^\mu} \qquad \textnormal{for}\,\, 0<x_1<1 \quad \textnormal{with} \,\,\,\, \mu>0,
\label{campo_W}
\end{equation}
vanishing for $x\geq 2$   and smooth for any $x_1>0$. We want to contrast this effect by means of a magnetic field, keeping the plasma outside the half-space.
We choose the magnetic field of the form 
\begin{equation}
B(x)= \left(0,0, h(x_1)\right)
\label{campo_B}
\end{equation}
with 
\begin{equation}
h(x_1) = \frac{1}{x_1^\tau}
\label{campo_B2}
\end{equation}
for $ 0<x_1<1$,     $ h(x_1)$  smooth for any $x_1>0$  and vanishing for $x\geq 2$  (such assumption, as for the field $U$, is made for
mathematical simplification only, but it could be removed).
We will show that,  with a suitable $\tau > \mu$, the plasma cannot penetrate the region $x_1 \leq 0.$
 In other words, for any attractive external field of the form (\ref{campo_W}) we can find
 a magnetic field of the form (\ref{campo_B2}) in order to obtain the shield effect.
We assume that 
\begin{equation}
\frac{\mu+1}{\tau-1}<\frac49, \qquad \textnormal{that is} \qquad
 \tau > \frac94 \mu + \frac{13}{4}.
 \label{rel_taumu}
 \end{equation}
A simple physical model which motivates a field of the form (\ref{campo_W}), in 
the coulomb case, could be to put in the plane $x_1=0$ some negative point charges, fixed at some points, whose effect is
similar to a potential of the form (\ref{campo_W}) for $\mu=1$, with the addition of transverse
components (with respect to the $x_1$-axis) of the force. We have used the potential  (\ref{campo_W})
for simplicity,  showing in the Appendix the confinement property in case of a negative charge  fixed at the origin.

\medskip

The main difficulty we face in this paper is concerned with the fact that, due to the presence of the external attracting field, the total energy of the system is not positive definite in sign. This problem, differently from preceding works  \cite{KRM16}, \cite{JSP17} and \cite{CCM18}, forces us to consider integrable space distributions of particles for the Vlasov-Poisson plasma. Even in this simplified case, the energy conservation does not give an immediate a priori bound on the total kinetic energy of the system which, as it is well known, would imply a bound on some $L^p$ norm of the spatial density. This is due to both the attractiveness and the divergence of the external field. Our plan is to deduce some estimates for a cutoff problem, in which the particle velocities are assumed to be initially bounded, and then to prove the existence of the time evolution of the system, by taking the limit as the cutoff goes to infinity. 
In conclusion, the new Vlasov-Poisson system differs from the ones considered in previous papers (see \cite{KRM16}, \cite{JSP17} and \cite{CCM18})
 for the presence of a new term, that changes the 
energy (no more positive definite). We show that we are able to control this effect.

The plan of the paper is the following. 
Section 2 is devoted to this aim, Section 3 contains our main estimate on the self-induced electric field and finally in Section 4 some technical lemmas are proven.
\bigskip

Let   $x=(x_1,x_2,x_3)$,   $v=(v_1,v_2,v_3)$, and $f(x,v,t)$  be the distribution function of the plasma in the phase space 
$(x,v)$ at time $t$, with finite   electric charge. It evolves via the equations of motion
\begin{equation}
\label{Eq}
\left\{
\begin{aligned}
&\partial_t f(x,v,t) + v\cdot \nabla_x f(x,v,t) \, + \\
&\left[E(x,t)+v \wedge B(x)-\nabla U(x)\right] \cdot \nabla_v f(x,v,t)=0 \\
&E(x,t)=\int_{\R^3 } \frac{x-y}{|x-y|^3} \, \rho(y,t) \, dy     \\
&\rho(x,t)=\int_{\R^3} f(x,v,t) \, dv \\
&f(x,v,0)=f_0(x,v)\geq 0,  \qquad  x\in {\R^3 },  \qquad v\in\R^3,  
\end{aligned}
\right.
\end{equation}
where  $\rho$ is the spatial density and  $E$ is the electric field generated by the plasma.

\noindent System (\ref{Eq}) shows that $f$ is time-invariant along the solutions of the so called characteristics equations:
\begin{equation}
\label{ch} 
\begin{cases}
\dsp  \dot{X}(t)= V(t)\\
\dsp  \dot{V}(t)=  E(X(t),t)+V(t) \wedge B(X(t)) - \nabla U(X(t))\\
\dsp (X(0), V(0))=(x,v)  \\
\dsp f(X(t), V(t), t) = f_0(x,v)
 \end{cases}
\end{equation}
where we  use the simplified notation
\begin{equation}
 \label{2.8}
(X(t),V(t))= (X(t,t',x,v),V(t,t',x,v)) 
 \end{equation}
 to represent a characteristic at time $t$ passing at time $0\leq t'<t$ through the point $(x,v)$. 
 It is well known that along (\ref{ch})
the partial differential equation (\ref{Eq}) transforms into an ordinary differential equation, hence
a result of existence and uniqueness of solutions to (\ref{ch}) implies the same result for solutions to (\ref{Eq}),
with regularity properties depending on the regularity of $f_0$.
Since $f$ is time-invariant along this motion,
 \begin{equation}
 \label{2.9}
\| f(t)\|_{L^\infty}= \| f_0 \|_{L^\infty}
 \end{equation}
 moreover this dynamical system preserves the measure of the phase space (Liouville's theorem)
 and the $L^1$ norm of $f$, that is
 \begin{equation}
 \int  f(x,v,t) \, dx dv  = \int  f_0(x,v) \, dx dv = \textnormal{total charge}.
 \end{equation}

\noindent The previous dynamical system has been studied in several papers;  for the
main classical results on the Vlasov-Poisson system in three dimension see  \cite{LP, Pf},
for a review see \cite{G},
while recent results are quoted in \cite {JSP17}. 

\noindent In the following we will concentrate on the study of system (\ref{ch}).

\medskip

\noindent Denote by $\Lambda= \{ x \in {\mathbb{R}}^3:  x_1>A>0 \}$, where the constant $A$ can
be taken arbitrarily small;  we will prove the following 
\begin{theorem}
Let us fix an arbitrary positive time $T$. Let  $f_0(x,v)$ be a bounded, integrable function with support in
$\Lambda \times {\mathbb{R}}^3$,
 satisfying the following bound in the velocities
\begin{equation}
\label{exp}
0\leq f_0(x,v) \leq  C_1 e^{-\lambda v^2} 
\end{equation}
for some positive constants $C_1$,    $\lambda$.

\noindent Then there exists  a solution to system  (\ref{ch}) on  $[0,T]$  such that
$
X_1(t)>0
$
  for any 
$(x,v) \in \Lambda \times {\mathbb{R}}^3$.

Moreover there exist  positive constants $C_2$,   $\lambda_1$,    such that, for any $t\in [0, T]$,
$f(x,v,t)$ remains bounded and integrable in $(x, v)$,  and
\begin{equation}
\label{exp 1}
0\leq f(x,v,t) \leq C_2 e^{-\lambda_1 v^2}.
\end{equation}
This solution is unique in the class of bounded, integrable in $(x,v)$, functions satisfying (\ref{exp 1}).

\label{th_01}
\end{theorem}

\bigskip

\noindent The proof will be given in the next Section. It could be adapted to different potentials $U$ and geometries of the  boundary  (which here is $x_1=0$).

\section{Proof of the Theorem}
\label{sec_proof}

\subsection{The partial dynamics}

\medskip

We introduce the partial dynamics, that is a sequence of differential systems in which the initial densities have compact velocity support: 
\begin{equation}
\label{chN}
\begin{cases}
\dsp  \dot{X}^{N}(t)= V^{N}(t)\\
\dsp  \dot{V}^{N}(t)= E^{N}(X^{N}(t),t)+V^{N}(t)\wedge B(X^{N}(t)) -\nabla U(X^N(t)) \\
\dsp ( X^{N}(0), V^{N}(0))=(x,v) \\
f_0^{N}(x,v) = f_0(x,v) \, \chi (|v| < N) \\
 \end{cases}
\end{equation}
where $\chi(\cdot)$ is the characteristic function of the set $(\cdot)$,

$$
E^{N}(x,t)=\int_{\R^3} \frac{x-y}{|x-y|^3} \ \rho^{N}(y,t) \, dy, 
$$
$$
\rho^{N}(x,t)=\int_{\R^3} f^{N}(x,v,t) \, dv
$$
and
\begin{equation}
f^{N}(X^{N}(t),  V^{N}(t),  t) = f_0^{N}(x,v).
\label{liouv}
\end{equation}
The total energy of the system is
\begin{equation}
\begin{split}
 {\mathcal{E}} = \,\frac12 \int dx\, &\left[\int dv \, |v|^2 f^N( x, v,t)\,+
 \rho^N( x,t)\int dy \, 
\frac{ \rho^N ( y,t)}{  | x- y|} \right] \\
& \, + \int dx \,U(x) \, \rho^N ( x,t) 
\end{split} 
\end{equation}
which is constant in time and finite, by virtue of the initial conditions. 
From this  we obtain, 
\begin{equation}
{\mathcal{K}}^N(t) := \frac12 \int dx\, \int dv \, |v|^2 f^N( x, v,t) \leq {\mathcal{E}} - \int dx  \,U(x) \,  \rho^N(x,t)
\label{kinet0}
\end{equation}
which will provide a bound for the kinetic energy useful in the sequel.

From now on all constants appearing in the estimates will be positive and 
possibly depending on the initial data and $T$, but not on $N$.
They will be denoted by $C$,  will be possibly different from line to line,
 and some of them will be numbered in order to be quoted elsewhere in the paper. 

We introduce the maximal velocity of a characteristic

 \begin{equation}
\label{maxV}
{\mathcal{V}}^N(t)=\max\left\{C_3,\,\sup_{s\in[0,t]}\sup_{(x,v)\in \Lambda \times b(N)}|V^N(s)|\right\}
\end{equation}
where  $b(N) = \{ v \in \R^3: |v| < N \}$,    $C_3$ is a constant that will be chosen large enough, and the maximal displacement  is defined by
\begin{equation}
R^N(t) = 1 + \int_0^t {\mathcal{V}}^N(s) \, ds.
\label{g}
\end{equation}

We premise the following fundamental result on the partial dynamics, which will be proved in Section \ref{sec_estE}.
\begin{proposition}
There exists $\gamma < 2/3$ such that
\begin{equation}
\int_0^t|E^N(X^N(s),s)|ds\leq C_4 \,{\mathcal{V}}^N(t)^\gamma.
\label{field_eq}
\end{equation}
\label{field}
\end{proposition}
As a consequence, the following holds:                                  
\begin{corollary}
\begin{equation}
{\mathcal{V}}^N(T) \leq C N
\label{V^N}
\end{equation}
\begin{equation}
|B(X^N(t))|\leq CN^{\frac{\tau}{\tau -1}}
\label{B}
\end{equation}
\begin{equation}
|\nabla U(X^N(t))| \leq C N^{\frac{\mu +1}{\tau -1}}
\label{nablaW}
\end{equation}
\begin{equation}
\rho^N(x, t) \leq C N^{3\gamma'} 
\label{rho_t}
\end{equation}
\begin{equation}
{\mathcal{K}}^N(t) \leq C N^{\frac{\mu}{\tau -1}}
\label{kinet}
\end{equation}
where  $\gamma'=\max \{ \gamma,  \frac{\mu+1}{\tau-1} \}$, 
with $\gamma$ the exponent in (\ref{field_eq}),  $\tau$ appearing in the definition of $B$  and
$\mu$ in the definition of $U$, and  satisfying  (\ref{rel_taumu}).
\label{coro}
\end{corollary}
\begin{proof}
To prove (\ref{V^N}) we observe that the external Lorentz force does not affect the modulus of the particle velocities, since
\begin{equation}
\frac{d}{dt}\frac{|V^N(t)|^2}{2}=V^N(t) \cdot \left[ E^N(X^N(t), t)
-\nabla U(X^N(t)) \right] \label{magn}
\end{equation}

hence
\begin{equation}
\begin{split}
|V^N(t)|^2=&|v|^2+2\int_0^tV^N(s) \cdot  E^N(X^N(s), s)\,ds  \\
& -2 \int_0^t V^N(s) \cdot \nabla U(X^N(t)) \, ds.
\end{split}
\label{v^2}
\end{equation}
Writing  the second component of system (\ref{chN}) we have
\begin{equation}
\begin{split}
\dot V_2^N(t) =& -  V_1^N(t) \,  h\left(X_1^N(t)\right) + E_2^N(X^N(t),t)  \\
=& -\frac{d}{dt} {\mathcal{H}}\left(X_1^N(t)\right)
+ E_2^N(X^N(t),t)
\end{split}
\label{compo2}
\end{equation}
where  ${\mathcal{H}}$ is a primitive of $h$.   Integrating in time we obtain
\begin{equation}
\begin{split}
V_2^N(t)-V_2^N(0) =& \, -{\mathcal{H}}\left(X_1^N(t)\right)
+ {\mathcal{H}}\left(X_1^N(0)\right)  \\
&+ \int_0^t  E_2^N(X^N(s),s) \, ds 
\end{split}
\label{shield}
\end{equation}
and, by Proposition \ref{field}, considering that ${\mathcal{H}}\left(X_1^N(0)\right)$ is finite by the assumption 
given in Theorem \ref{th_01} (the distance of the plasma from the border $x_1=0$ is initially greater
than $A>0$) we get
\begin{equation}
|{\mathcal{H}}\left(X_1^N(t)\right)| \leq C {\mathcal{V}}^N(t)
\end{equation}
from which, using the form of the magnetic field which holds near the border (\ref{campo_B2}),
\begin{equation}
\frac{1}{\left(X_1^N(t)\right)^{\tau-1}} \leq C {\mathcal{V}}^N(t).
\label{min_dist}
\end{equation}
From (\ref{min_dist}) and (\ref{campo_W}) we obtain
\begin{equation}
|\nabla U(X^N(t))| \leq \frac{C}{\left(X_1^N(t)\right)^{\mu+1}} \leq C \left[{\mathcal{V}}^N(t)\right]^{\frac{\mu+1}{\tau-1}}
\label{nablaW'}
\end{equation}
and this last estimate,  Proposition \ref{field}, and the choice of the initial data $v\in b(N)$,  imply for (\ref{v^2})
\begin{equation}
|V^N(t)|^2 \leq N^2+C \left[{\mathcal{V}}^N(t)\right]^{\gamma+1} +C \left[{\mathcal{V}}^N(t)\right]^{\frac{\mu+1}{\tau-1}+1}.
\end{equation} 
Hence, since $\gamma+1<2$ and $\frac{\mu+1}{\tau-1}+1<2$, by taking the $\sup_{t\in [0,T]}$ we obtain the thesis.

The proof of (\ref{B}) follows easily from (\ref{V^N}), (\ref{min_dist}) and the definition of the magnetic field;
 the proof  (\ref{nablaW}) comes from (\ref{V^N}) and  (\ref{nablaW'}).

 Now we prove (\ref{rho_t}). We make the (invertible) change of variables $(x,v)\to (\bar{x}, \bar{v})$
 $$
 (\bar{x}, \bar{v})=\left(X^N(x,v,t),V^N(x,v,t)\right),$$
then, by using the invariance of the density along the characteristics, we have 
\begin{equation*}
\rho^N(\bar{x}, t) =\int f(\bar{x}, \bar{v},t)\ d\bar{v}=\int f_0(x,v)\ d\bar{v}
\end{equation*}
where in the last integral $(x,v)$ are functions of $(\bar{x}, \bar{v})$  accordingly to the former transformation.
We notice that, putting
$$
\widetilde{V}^N(t)=\sup_{0\leq s\leq t}|V^N(s)|,
$$
 from (\ref{v^2}), Proposition  \ref{field},  (\ref{V^N})  and (\ref{nablaW})  it follows
\begin{equation}
\begin{split}
|v|^2\geq |V^N(t)|^2-C_5\widetilde{V}^N(t) N^{\gamma'}
\end{split}\label{VN}
\end{equation}
remembering the definition of $\gamma'=\max \{ \gamma,  \frac{\mu+1}{\tau-1} \}$.
Hence, we decompose the integral as follows
\begin{equation}
\begin{split}
&\rho^N(\bar{x}, t)\leq\\&\int_{\widetilde{V}^N(t)\leq 2C_5 N^{\gamma'}} f_0(x,v)\ d\bar{v}+C \int_{\widetilde{V}^N(t)> 2C_5 N^{\gamma'}} e^{-\lambda |v|^2}\ d\bar{v}\\& \leq   CN^{3\gamma'}+  C \int_{\widetilde{V}^N(t)> 2C_5 N^{\gamma'}} e^{-\lambda |v|^2}\ d\bar{v},   \label{f_0}
\end{split}
\end{equation}
since  by definition $|\bar{v}|\leq \widetilde{V}^N(t).$

\noindent By (\ref{VN}), if $\widetilde{V}^N(t)> 2C_5 N^{\gamma'},$ then $$|v|^2\geq |V^N(t)|^2-\frac{[ \widetilde{V}^N(t)]^2}{2}.$$
Since the inequality holds for any $t\in [0,T],$ it holds also at the time in which $V^N$ reaches its maximal value over $[0,t],$ that is 
\begin{equation}
|v|^2 \geq  [\widetilde{V}^N(t)]^2-\frac{[ \widetilde{V}^N(t)]^2}{2}= \frac{[\widetilde{V}^N(t)]^2}{2}\geq \frac{| {V}^N(t)|^2}{2}.
\label{shield2}
\end{equation}
Hence from (\ref{f_0}) it follows
\begin{equation}
\begin{split}
&\rho^N(\bar{x}, t)\leq CN^{3\gamma'}+C \int e^{-\frac{\lambda}{2}|\bar{v}|^2}\ d\bar{v} \leq CN^{3\gamma'}.\end{split}
\end{equation}
We come now to the proof of (\ref{kinet}).  By (\ref{kinet0}) and (\ref{campo_W}), considering the minimum distance from the border $x_1=0$
as expressed by (\ref{min_dist}), we get
$$
{\mathcal{K}}^N(t) \leq C \left( \mathcal{V}^N(t)\right)^{\frac{\mu}{\tau-1}}  \int dx \, \rho^N(x,t) \leq 
C \left( \mathcal{V}^N(t)\right)^{\frac{\mu}{\tau-1}}
$$
since the integral of $\rho^N$ is the total charge, which is initially finite and is conserved.
Hence by (\ref{V^N}) we obtain (\ref{kinet}).

\end{proof}

\bigskip

We give now a first estimate on the electric field $E,$ in terms of the maximal velocity. We will make use of it in the proof of Proposition \ref{field} in Section \ref{sec_estE}.
\begin{proposition}
\begin{equation}
|E^N(x,t)|\leq C_6\left[{\mathcal{V}}^N(t)\right]^{\frac43 +  \frac{\mu}{3(\tau -1)}}.
\label{C1}
\end{equation}
\label{prop2}\end{proposition}

\begin{proof} 
We  premise an estimate on the spatial density: 
\begin{equation}
\int dx \, \rho^N(x,t)^{\frac53}\leq C \, {\mathcal{K}}^N(t) \leq C \left[{\mathcal{V}}^N(t)\right]^{\frac{\mu}{\tau -1}}.
\label{lem1}
\end{equation}

In fact:
\begin{equation*}
\begin{split}
\rho^N(x,t)&\leq \int_{|v|\leq a} dv f^N(x,v,t)+\frac{1}{a^2}\int_{|v|>a}dv\ |v|^2 f^N(x,v,t)\leq \\&C a^3+\frac{1}{a^2}\int dv\ |v|^2 f^N(x,v,t).
\end{split}
\end{equation*}
By minimizing over $a,$  taking the power $\frac53$ of both members and integrating in space,
recalling (\ref{kinet}), we get (\ref{lem1}).

Then we have
\begin{equation}
\begin{split}
|E^N(x,t)|&\leq \int  \rho^N(y,t) \, \frac{1}{|x-y|^2} \, dy \\
&\leq \int_{|x-y|\leq \varepsilon} \rho^N(y,t) \, \frac{1}{|x-y|^2} \, dy +
\int_{|x-y| > \varepsilon} \rho^N(y,t) \, \frac{1}{|x-y|^2} \, dy \\
&\leq C \,\varepsilon\, \|\rho^N(t)\|_{L^{\infty}}  + \left(\int \rho^N(y,t)^\frac53 dy \right)^\frac35
\left( \int_{|x-y|>\varepsilon} \frac{1}{|x-y|^5}\, dy  \right)^\frac25  \\
&\leq C\, \varepsilon \left[{\mathcal{V}}^N(t) \right]^3 + C \left[{\mathcal{V}}^N(t)\right]^{\frac{3\mu}{5(\tau -1)}}
 \varepsilon^{-\frac45}.
\end{split}
\end{equation}
Minimizing with respect to $\varepsilon$,
$$
|E^N(x,t)| \leq C \left[{\mathcal{V}}^N(t) \right]^{\frac43+  \frac{\mu}{3(\tau -1)}}
$$
and we observe that the term $\frac{\mu}{3(\tau -1)}$ is sufficiently small  (smaller than $\frac15$),
which  will be used later in (\ref{para}).

\end{proof}

\subsection{Convergence of the partial dynamics}
\label{conv_pd}

We repeat here for completeness the proof of the convergence of the partial dynamics as $N\to\infty$ 
 in close analogy to what was done in \cite{JSP17}.
 We denote by $\Lambda= \{ x \in {\mathbb{R}}^3:  x_1>A>0 \}$  the spatial support of $f_0$ and  by $b(N)$   the ball in $\mathbb{R}^3$ of center $0$ and radius $N.$ We fix a couple $(x,v)\in  \Lambda\times b(N)$,
 and we consider $X^N(t)$ and $X^{N+1}(t),$ that is the time evolved characteristics, both starting from this initial condition, in the  different dynamics relative to the initial distributions $f_0^N$ and $f_0^{N+1}.$
 
  We set
\begin{equation}
\delta^N(t) =\sup_{(x,v)\in \Lambda\times b(N)} |X^N(t)-X^{N+1}(t)|
\label{delta^N}
\end{equation}
\begin{equation}
\eta^N(t)=\sup_{(x,v)\in \Lambda\times b(N)} |V^N(t)-V^{N+1}(t)|
\end{equation}
and 
\begin{equation}
\sigma^N(t)=\delta^N(t)+\eta^N(t).
\end{equation}
Our goal is to prove the following estimate
\begin{equation}
\sup_{t\in[0,T]}\sigma^N(t)\leq C^{-N^c} \label{fff} 
\end{equation}
with $c$ a positive exponent and $C>1$.  Once we get this estimate, the proof of the Theorem is accomplished. Indeed, (\ref{fff}) implies that the sequences $X^N(t)$ and $V^N(t)$ are Cauchy sequences, uniformly on  $[0,T].$ 
 Therefore, for any fixed $(x,v)$, they are bounded uniformly in $N$.
  
Indeed, let us fix an initial datum $(x,v)$ and choose a positive integer $N_0$ such that $N_0=\hbox{Intg}(|v|+C),$ where $\hbox{Intg}(a)$ represents the integer part of the number $a$ and $C$ is assumed sufficiently large. Then, for any $N>N_0,$ by (\ref{fff}) and (\ref{V^N}) in Corollary \ref{coro}, we have 
\begin{equation}
\begin{split}
|V^N(t)-v|\leq\,& |V^{N_0}(t)-v|+\sum_{k=N_0}^{N-1}\eta^k(t)\leq\\& |V^{N_0}(t)-v|+C\leq |v| +CN_0 \end{split}
\end{equation}
which implies, by the choice of $N_0,$
\begin{equation}
|V^N(t)|\leq C(|v|+1)\label{N_0}
\end{equation}
and hence
\begin{equation}
|X^N(t)|\leq |x|+ C\left(|v|+1\right).\label{N_1}
\end{equation}
Since for any fixed $(x,v)$ the sequences $\{X^N(t)\},$ $\{V^N(t)\}$, and consequently $\{f^N(t)\}$, are  Cauchy sequences, they converge, as $N\to \infty,$ to some limit functions, which we call $X(t)$,  $V(t)$,  $f(x,v,t)$.

Property (\ref{N_0})  allows to show that the functions $f^N(t)$ enjoy property (\ref{exp 1}), with  
constants  $\bar C$ and  ${\lambda_1}$ independent of $N$. 
Indeed by (\ref{N_0}),  we observe that
\begin{equation}\begin{split}
f^N(X^N(t), V^N(t),t)=&f^N_0(x,v)\leq  C_1 e^{-\lambda |v|^2} \\& \leq  C_2 e^{-{\lambda_1}|V^N(t)|^2}.\label{G2}
\end{split}\end{equation}
Moreover, 
$$
\int f^N(x,v,t) \, dxdv < C
$$
since the total charge is initially  finite and is conserved.
This implies that also the limit function $f(x,v,t)$ is integrable in $(x,v)$ and  satisfies  property (\ref{exp 1}).

  The uniqueness of the solution and the fact that it  satisfies
 system (\ref{ch})  follow directly from what already done in Section 3 of \cite{CCM18}, where the plasma evolves in the whole  ${\mathbb{R}}^3$ in case of infinite charge and gaussian velocities. 
 
 The fact that for any $(x,v)$ in the support of $f_0$ it results $X_1(t)>0$ follows from (\ref{shield}) combined with
 (\ref{shield2}). In fact, writing (\ref{shield}) for the limit solution, 
  using (\ref{shield2}) to bound $V_2(t)$ on the left hand side of (\ref{shield}), and using  the  electric field as expressed in  (\ref{Eq}) (which turns out to be bounded, using the properties of $f(x,v,t)$), we immediatly obtain $X_1(t)>0$ for any $t\in [0,T]$.  
   
\bigskip

\noindent $\mathbf{Proof\, of \,estimate\, (\ref{fff})}$.

\noindent  We have
\begin{equation}
\begin{split} 
&\delta^N(x,v,t)=|X^N(t)-X^{N+1}(t)| =\\&  \, \bigg| \int_0^t dt_1 \int_0^{t_1} dt_2 \, \Big[E^N\left(X^N(t_2), t_2\right)
+V^N(t)\wedge B\left(X^{N}(t_2)\right)   \\
&- \nabla U\left(X^N(t_2) \right)                  
  -  E^{N+1}\left(X^{N+1}(t_2), t_2\right)  \\
&-V^{N+1}(t)\wedge B\left(X^{N+1}(t_2)\right) +\nabla U\left(X^{N+1}(t_2) \right) \Big] \bigg| \leq   \\
&\int_0^t dt_1 \int_0^{t_1} dt_2 \, \left[ \mathcal{G}_1(x,v,t_2) + \mathcal{G}_2(x,v,t_2) + \mathcal{G}_3(x,v,t_2)
+\mathcal{G}_4(x,v,t_2)  \right]
\end{split}
\label{iter_}
\end{equation}
where
\begin{equation}
\mathcal{G}_1(x,v,t) = \left|E^N\left(X^N(t), t\right)
-E^N\left(X^{N+1}(t), t\right)\right|
\end{equation}
\begin{equation}
\mathcal{G}_2(x,v,t) = \left|E^N\left(X^{N+1}(t), t\right)
-E^{N+1}\left(X^{N+1}(t), t\right)\right|
\end{equation}
\begin{equation}
\mathcal{G}_3(x,v,t) = \left|V^N(t)\wedge B\left(X^N(t)\right)
-V^{N+1}(t)\wedge B\left(X^{N+1}(t)\right)\right|
\end{equation}
\begin{equation}
\mathcal{G}_4(x,v,t) = \left|- \nabla U\left(X^N(t) \right) 
+\nabla U\left(X^{N+1}(t) \right) \right|. 
\end{equation}

To estimate the first term $\mathcal{G}_1$
we have to prove a quasi-Lipschitz property for $E^N.$ We consider the difference $|E^N(x,t)-E^N(y,t)|$ at two generic points $x$ and $y,$ (which will be in our particular case $X^N(t)$ and $X^{N+1}(t)$) and set: 
\medskip

\begin{equation*}
\mathcal{G}_1'(x,y,t)= |E^N(x,t)-E^N(y,t)|\chi (|x-y|\geq 1)
\end{equation*}
and
\begin{equation*}
\mathcal{G}_1''(x,y,t)= |E^N(x,t)-E^N(y,t)|\chi (|x-y|< 1).
\end{equation*}
 Hence 
\begin{equation}
\mathcal{G}_1(x,v,t)=\mathcal{G}_1'(X^N(t),X^{N+1}(t),t)+\mathcal{G}_1''(X^N(t),X^{N+1}(t),t). \label{F}\end{equation}
\medskip
By Proposition \ref{prop2} and Corollary \ref{coro} we have
\begin{equation}
\begin{split}
\mathcal{G}_1'(x,y,t)\leq&\ |E^N(x,t)|+|E^N(y,t)|\leq C N^{\frac43+\frac{\mu}{3(\tau-1)}}\leq
 CN^{3\gamma} \leq CN^{3\gamma}|x-y| \label{b1}
\end{split}
\end{equation}
which can be realized suitably adapting $\gamma<\frac23$ (introduced in (\ref{field_eq})).

Let us now estimate the term $\mathcal{G}_1''$. 

We put $\zeta=\frac{x+y}{2}$ and consider the sets: 

\begin{equation}
\begin{split}
&S_1=\{z:|\zeta-z|\leq 2|x-y|\}, \\& S_2=\Big\{z:2|x-y|\leq |\zeta-z|\leq\frac{4}{|x-y|}\Big\}\\& S_3=\Big\{z:|\zeta-z|\geq \frac{4}{|x-y|}\Big\}.\label{S'}
\end{split}
\end{equation}
We have:
\begin{equation}
\begin{split}
&\mathcal{G}_1''(x,y,t)\leq  \int_{S_1\cup S_2\cup S_3} \left|\frac{x-z}{|x-z|^3}-\frac{y-z}{|y-z|^3}\right| \rho^N(z,t) \,dz.
\end{split} \label{s}
\end{equation}

By (\ref{rho_t}) and the definition of $\zeta$ we have
\begin{equation}
\begin{split}
&\int_{S_1} \left| \frac{x-z}{|x-z|^3}-\frac{y-z}{|y-z|^3}\right| \rho^N(z,t) \,dz\leq\\& CN^{3\gamma'}   \int_{S_1} \left|\frac{1}{|x-z|^2}+\frac{1}{|y-z|^2}  \right|\,dz\leq C N^{3\gamma'}  |x-y|.
\end{split}
\label{Ai1}\end{equation}
Let us pass to the integral over the set $S_2.$ By the Lagrange theorem
applied to each component   $i=1, 2, 3$,  of $E^N(x,t)-E^N(y,t)$, there exists $\xi_z$  such that $\xi_z=\kappa x +(1-\kappa)y$ and $\kappa\in [0,1]$ (depending on $z$), for which
\begin{equation} 
\begin{split}
\int_{S_2} \left|\frac{x_i-z_i}{|x-z|^3}-\frac{y_i-z_i}{|y-z|^3}\right|& \rho^N(z,t) dz\leq C |x-y|\int_{S_2} \frac{\rho^N(z,t)}{|\xi_z-z|^3}\,dz\\&\leq  CN^{3\gamma'} |x-y|\int_{S_2} \frac{1}{|\xi_z-z|^3}\,dz.
\end{split}\label{D'}
\end{equation}
Since in $S_2$ it results $|\xi_z-z|\geq \frac{|\zeta-z|}{2}$, we have
$$
\int_{S_2} \frac{1}{|\xi_z-z|^3}\,dz \leq 8\int_{S_2} \frac{1}{|\zeta-z|^3}\,dz \leq C (|\log |x-y||+1)
$$

%where $S_2'=\left\{z: |x-y|<|\xi_z-z|\leq \frac{5}{|x-y|}\right\}.$
\noindent and combining the results for the three components
\begin{equation}
\begin{split}
&\int_{S_2} \left|\frac{x-z}{|x-z|^3}-\frac{y-z}{|y-z|^3}\right| \rho^N(z,t) \,dz\leq  CN^{3\gamma'} 
|x-y|\, (|\log |x-y||+1). \label{D}
\end{split}
\end{equation}

For the last integral over $S_3,$ again by the Lagrange theorem, we have for some $\xi_z=\kappa x +(1-\kappa)y$ and $\kappa\in [0,1],$  
\begin{equation}
 \int_{ S_3} \left|\frac{x_i-z_i}{|x-z|^3}-\frac{y_i-z_i}{|y-z|^3}\right| \rho^N(z,t)\, dz\leq C |x-y| \int_{ S_3} \frac{\rho^N(z,t)}{|\xi_z-z|^3}\,dz.\label{S}
 \end{equation}
Notice that, since $z\in S_3$ and $|x-y|<1$ by definition of $\mathcal{G}_1''$, it is 
%$$
%|\xi_z - z|\geq  |\zeta - z| - |x-y| \geq \frac{3}{|x-y|}> 1,
%$$
$$
|\xi_z - z|\geq  |\zeta - z| - |x-y|,
$$
and 
$$
\frac{1}{|\xi_z - z|} \leq \frac{1}{|\zeta - z| - |x-y|} = \frac{1}{\frac34|\zeta - z| +\frac14 |\zeta-z|- |x-y|}
\leq \frac43 \frac{1}{|\zeta-z|},
$$
 then we have, by (\ref{lem1}) and H\"older's inequality,
\begin{equation}
\begin{split}
&\int_{ S_3} \frac{\rho^N(z,t)}{|\xi_z-z|^3}\,dz \leq C\int_{|\zeta-z|\geq 4 } \frac{\rho^N(z,t)}{|\zeta-z|^3}\,dz \leq \\
& C N^{\frac35 \frac{\mu}{\tau-1}} \left[\int_{|\zeta-z|\geq 4 } \left(\frac{1}{|\zeta-z|^3}\right)^\frac52\,dz\right]^\frac25
\leq C N^{\frac35 \frac{\mu}{\tau-1}}.
\label{hold}
\end{split}
\end{equation}
Using this estimate in (\ref{S}) and going back to (\ref{s}), by (\ref{Ai1}) and (\ref{D}) we have proved that
\begin{equation}
\mathcal{G}_1''(x,y,t)\leq CN^{3\gamma'}  |x-y|  \, (|\log |x-y||+1). \label{f1}
\end{equation}

\bigskip
 In conclusion, recalling (\ref{F}), estimates (\ref{b1}),  (\ref{f1}), and the definition of $\delta^N(t)$ (\ref{delta^N}) show that
\begin{equation}
\mathcal{G}_1(x,v,t)\leq  CN^{3\gamma'}  \left(1+|\log \delta^N(t)|\right)\delta^N(t). \label{f1'}
\end{equation}

\medskip

We now concentrate on the term $\mathcal{G}_2$. 
Putting $\bar{X}=X^{N+1}(t),$  we have:
\begin{equation}
\mathcal{G}_2(x,v,t)\leq \mathcal{G}_2'(\bar X,t)+\mathcal{G}_2''(\bar X,t),
\label{I_2}
\end{equation}
where
\begin{equation}
\mathcal{G}_2'(\bar X,t)=\left|\int_{|\bar X - y|\leq 2\delta^N(t)} \frac{  \rho^N(y,t)-\rho^{N+1}(y,t)}{|\bar X-y|^2} \, dy \right|
\end{equation}
\begin{equation}
\mathcal{G}_2''(\bar X,t)=\left|\int_{2\delta^N(t)\leq |\bar X - y| } \frac{  \rho^N(y,t)-\rho^{N+1}(y,t)}{|\bar X-y|^2} \, dy \right|.
\end{equation}
By estimate (\ref{rho_t}) we get
\begin{equation}
\mathcal{G}_2'(\bar X,t)\leq \int_{|\bar X - y|\leq 2\delta^N(t)} \frac{  \rho^N(y,t)+\rho^{N+1}(y,t)}{|\bar X-y|^2} \, dy \leq CN^{3\gamma'}  \delta^N(t).\label{F2'}
\end{equation}
Now we estimate the term $\mathcal{G}_2''.$ By the invariance of $f^N(t)$ along the characteristics, making a change of variables, we decompose the integral as follows:
\begin{equation}
\begin{split}
\mathcal{G}_2''(\bar X,t)=  \left| \int_{S^N(t)} d y \, dw\,\frac{f_0^N( y, w)}{|\bar X-  Y^N(t)|^2}\ -\int_{S^{N+1}(t)}dy \, dw\,\frac{f_0^{N+1}(y,w)}{|\bar X-Y^{N+1}(t)|^{2}} \right|
  \end{split}
\label{I_2''}
\end{equation}
where we put, for $i=N, N+1,$ $$(Y^i(t), W^i(t))= (X^i(y,w,t), V^i(y,w,t))$$ 
$$
S^{i}(t)=\{( y, w):2\delta^N(t)\leq|{\bar X}- Y^i(t)| \}.
$$
We notice that it is
\begin{equation} 
\mathcal{G}_2''(\bar X,t)\leq \mathcal{I}_1+\mathcal{I}_2+\mathcal{I}_3\label{l_i}
\end{equation}
where
\begin{equation}
\mathcal{I}_1= \int_{S^N(t)} d y  
 \int d w \left| \frac{1}{|\bar X -  Y^N(t)|^{2}}
-\frac{1}{|\bar X -  Y^{N+1}(t)|^{2}} \right| f_0^N( y,w),
\end{equation}
\begin{equation}
\mathcal{I}_2= \int_{S^{N+1}(t)} d y   
  \int d w \, \frac{ \left| f_0^N( y,w) - f_0^{N+1}( y, w) \right|}{|\bar X -  Y^{N+1}(t)|^{2}}  \, 
  \end{equation}
 \begin{equation}
  \mathcal{I}_3=
  \int_{S^N(t)\setminus S^{N+1}(t)}dy\int dw\,\frac{f_0^N(y,w)}{\left|\bar{X}-Y^{N+1}(t)\right|^2}.
  \end{equation}
We start by estimating $\mathcal{I}_1.$ 
  By the Lagrange theorem 
 \begin{equation}
\mathcal{I}_1\leq C \sup_{(y,w)\in \Lambda \times b(N)} |Y^{N}(t)- Y^{N+1}(t)| \int_{S^N(t)} dy\int dw \, \frac{ f_0^N(y,w)}{|\bar X - \xi^N(t) |^{3}} 
\end{equation}
where $\xi^N(t)= \kappa Y^{N}(t)+ (1-\kappa) Y^{N+1}(t)$   for a suitable $\kappa \in [0, 1]$.
By putting 
\begin{equation*}
\left(\bar{y}, \bar{w}\right)= \left((Y^N(t), W^N(t)\right)    
\end{equation*}
and
$$
{\bar{S}}^N(t)=\left\{ (\bar y, \bar w): (y, w)\in  S^N(t) \right\},
$$
 we get
\begin{equation}
\begin{split}
\mathcal{I}_1\leq & \, C  \delta^N(t)\int_{ {\bar{S}}^N(t)}d\bar{y}\int d\bar{w} \, \frac{ f^N(\bar{y},\bar{w}, t)}{|\bar X - \xi^N(t) |^{3}}.\end{split}\label{bar}
\end{equation}

If $(y,w)\in S^N(t)$ then\begin{equation*}
 |\bar X -  \xi^N(t)| > |\bar X - \bar{y}| - |\bar{y} -Y^{N+1}(t)|\geq|\bar X - \bar{y}|-\delta^N(t)\geq \frac{|\bar X - \bar{y}|}{2}\label{>>}
\end{equation*}
which, by (\ref{bar}), implies
\begin{equation}
\begin{split}
\mathcal{I}_1 \leq \ &C \delta^N(t)\int_{ \bar S^N(t)}d\bar{y}\int d\bar{w} \, \frac{ f^N(\bar{y},\bar{w}, t)}{|\bar X - \bar{y} |^{3}}\,=C \delta^N(t)\int_{\bar S^N(t)}d\bar{y}\, \frac{ \rho^N(\bar{y}, t)}{|\bar X - \bar{y} |^{3}}
\end{split}.
\end{equation}
Now we consider the sets 
$$
A_1= \left\{(\bar y,\bar w):2 \delta^N(t)\leq\left|\bar{X}-\bar{y}\right|\leq  \frac{4}{\delta^N(t)}\right\}
$$
$$
A_2= \left\{(\bar y,\bar w): 1\leq\left|\bar{X}-\bar{y}\right|\right\}. 
$$
Then it is $\bar S^N(t)\subset \bigcup_{i=1,2} A_i$ 
and
\begin{equation}
\mathcal{I}_1 \leq C \delta^N(t)\sum_{i=1}^2\int_{  A_i}d\bar{y}\frac{ \rho^N(\bar{y}, t)}{|\bar X - \bar{y} |^{3}}.
\end{equation}
We estimate the integral over $A_1$ as we did in (\ref{D'}), the one over $A_2$ is bounded as in (\ref{hold}). Hence we get
\begin{equation}
\begin{split}
\mathcal{I}_1  \leq CN^{3\gamma'}   \delta^N(t)\left[|\log \delta^N(t)|+1\right]\end{split}.    
\label{i1'}
\end{equation}

We estimate now the term $\mathcal{I}_2.$ By the definition of $f_0^i,$ for $i=N, N+1$,  it follows 
\begin{equation}
\begin{split}
\mathcal{I}_2 =
 \int_{S^{N+1}(t)} d y   
  \int_{} d w\,  \frac{  f_0^{N+1}( y, w) }{|\bar X -  Y^{N+1}(t)|^{2}} \chi\left(N\leq |w|\leq N+1\right).
\end{split} \label{AA} \end{equation}
We evaluate the integral over $S^{N+1}(t)$ by considering the sets   ($R^N(t)$ is defined in (\ref{g}))   
$$
B_1=   \left\{(y,w): |\bar{X}-Y^{N+1}(t)|\leq4R^{N+1}(t)\right\} $$
$$
 B_2=  \left\{(y,w):4R^{N+1}(t)\leq  |\bar{X}-Y^{N+1}(t)|\right\},
$$
so that 
\begin{equation}
S^{N+1}(t)\subset \bigcup_{i=1,2}B_i. \label{si}
\end{equation}
By putting 
\begin{equation*}
(\bar{y}, \bar{w})=\left((Y^{N+1}(t), W^{N+1}(t)\right),
\end{equation*}
 by (\ref{V^N}) it is $|\bar{w}|\leq CN,$ so that we have (recalling also (\ref{exp}))
\begin{equation}
\begin{split}
&\int_{B_1} d y   
  \int_{} d w\,  \frac{  f_0^{N+1}( y, w) }{|\bar X -  Y^{N+1}(t)|^{2}} \chi\left(N\leq |w|\leq N+1\right) \leq  \\
&C e^{-\lambda N^2}\int_{B_1} dy \int_{} dw\,\, \frac{ 1 }{|\bar X - Y^{N+1}(t)|^{2}} \chi\left(|w|\leq N+1\right)  \leq \\
&C e^{-\lambda N^2} \int_{|\bar{w}|\leq C N} d\bar{w}\int_{ |\bar{X}-\bar{y}|\leq 4R^{N+1}(t)} d\bar{y}\, \frac{1 }{|\bar X - \bar{y}|^{2}}\leq \\
&C e^{-\lambda N^2}N^3R^{N+1}(t)\leq C e^{-\lambda N^2}N^4 \leq C e^{-\frac{\lambda}{2} N^2}.
\label{g5}
\end{split}\end{equation}

For the integral over the set $B_2,$ we observe that, if $|\bar{X}-Y^{N+1}(t)|\geq 4R^{N+1}(t),$ then $|\bar{X}-y|\geq 3R^{N+1}(t)$ and  
\begin{equation*}
|\bar{X}-Y^{N+1}(t)|\geq |\bar{X}-y|-R^{N+1}(t)\geq \frac23|\bar{X}-y|.
\end{equation*} Hence  for a suitable $0<q<1$ chosen later,  using also  (\ref{exp}), 
\begin{equation}
\begin{split}
&\int_{B_2} d y   
  \int_{} d w\,  \frac{  f_0^{N+1}( y, w) }{|\bar X -  Y^{N+1}(t)|^{2}} \chi\left(N\leq |w|\leq N+1\right) = \\
&\int_{B_2} d y   
  \int_{} d w\,  \frac{  \left[f_0^{N+1}( y, w)\right]^q  \left[f_0^{N+1}( y, w)\right]^{1-q} }{|\bar X -  Y^{N+1}(t)|^{2}} \chi\left(N\leq |w|\leq N+1\right) \leq  \\  
 &C e^{-\lambda q N^2 }\int_{B_2} d y   
  \int_{} d w\,  \frac{    \left[f_0^{N+1}( y, w)\right]^{1-q} }{|\bar X -  Y^{N+1}(t)|^{2}} \chi\left(|w|\leq N+1\right) \leq  \\   
  &C e^{-\lambda q N^2 }\int_{|\bar{X}-y|\geq 3R^{N+1}(t)} d y   
  \int_{} d w\,  \frac{  \left[f_0^{N+1}( y, w)\right]^{1-q} }{|\bar X -  y|^{2}} \chi\left(|w|\leq N+1\right) \leq  \\
 &C e^{-\lambda q N^2 }\left(\int dydw \,  f_0^{N+1}( y, w)\right)^{1-q}
 \left( \int_{|\bar{X}-y|\geq 1} dydw \left[\frac{\chi\left(|w|\leq N+1\right)}{|\bar X -  y|^{2}}\right]^\frac1q  \right)^q    \\
 &\leq \, C e^{-\lambda q N^2 } N^{3q} \leq C e^{-\frac{\lambda}{2}  N^2 }
   \label{g6}
\end{split}\end{equation}
with $\frac12<q<\frac23$  in order to have a bound analogous to (\ref{g5}) and to assure the
 convergence of the spatial integral of $|\bar X -  y|^{-2/q}$, 
using also  the fact that the  total charge is finite.

\noindent Going back to (\ref{AA}), by (\ref{si}), (\ref{g5}) and (\ref{g6}) we have
\begin{equation}
\mathcal{I}_2\leq   Ce^{-\frac{\lambda}{2} N^2}.
\label{I}
\end{equation}
Finally, we estimate the term $\mathcal{I}_3.$ If $(y,w)\in S^N(t),$ then
$$
\left|\bar{X}-Y^{N+1}(t)\right|\geq \left|\bar{X}-Y^{N}(t)\right| -\delta^N(t)\geq \delta^N(t),
$$
so that
 \begin{equation}
\begin{split}
 & \mathcal{I}_3 =
  \int_{S^N(t)\setminus S^{N+1}(t)}dy\int dw\,\frac{f_0^N(y,w)}{\left|\bar{X}-Y^{N+1}(t)\right|^2}\leq\\&
\frac{1}{[\delta^N(t)]^2}  \int_{\left(S^{N+1}(t)\right)^c}dy\int dw\,f_0^N(y,w).
 \end{split} \end{equation}
 If $(y,w)\in \left(S^{N+1}(t)\right)^c,$ then $$\left|\bar{X}-Y^{N}(t)\right|\leq \left|\bar{X}-Y^{N+1}(t)\right|+ \delta^N(t)\leq 3\delta^N(t).$$ 
 Hence, by putting 
 \begin{equation*}
 (\bar{y}, \bar{w}) = \left(Y^N(t), W^N(t)\right)
 \end{equation*}
  we get
 \begin{equation}
\begin{split}
  \mathcal{I}_3\leq &
\,\frac{1}{[\delta^N(t)]^2}  \int_{\left|\bar{X}-\bar{y}\right|\leq 3\delta^N(t)} d\bar{y}\int d\bar{w}f^N(\bar{y}, \bar{w}, t)\leq\\&\frac{1}{[\delta^N(t)]^2}  \int_{ \left|\bar{X}-\bar{y}\right|\leq 3\delta^N(t)}d\bar{y}\,\rho(\bar{y},t)\leq \\& C\frac{N^{3\gamma'}}{[\delta^N(t)]^2}  \int_{ \left|\bar{X}-\bar{y}\right|\leq 3\delta^N(t)}d\bar{y}\,\leq CN^{3\gamma'}\,
\delta^N(t).\end{split}\label{I3}
  \end{equation}
Going back to (\ref{l_i}), by (\ref{i1'}), (\ref{I}) and (\ref{I3}) we have
\begin{equation}
\mathcal{G}_2''\leq C\left[N^{3\gamma'}  \delta^N(t)\left(|\log \delta^N(t)|+1\right)+e^{-\frac{\lambda}{2} N^2}\right], \label{F2''}
\end{equation}
so that, by (\ref{I_2}), (\ref{F2'}) and (\ref{F2''}), we get
\begin{equation}
\mathcal{G}_2\leq C\left[N^{3\gamma'} \delta^N(t)\left(|\log \delta^N(t)|+1\right)+e^{-\frac{\lambda}{2} N^2}\right].\label{f}
\end{equation}
Hence, going back to the estimate (\ref{f1'}) of the term $\mathcal{G}_1$ we have that
\begin{equation}
\mathcal{G}_1+\mathcal{G}_2\leq C\left[N^{3\gamma'} \delta^N(t)\left(|\log \delta^N(t)|+1\right)+e^{-\frac{\lambda}{2} N^2}\right].\label{g'}
\end{equation}

We need now a property, which  is easily seen to hold by convexity, that is, for any $s\in (0,1)$ and $a\in (0,1)$ it holds:
$$
 s(|\log s|+1)\leq s|\log a|+a.
 $$
Hence, for $\delta^N(t)< 1,$ we have, for any $a<1,$
$$
\mathcal{G}_1+\mathcal{G}_2\leq C\left[N^{3\gamma'}  \left( \delta^N(t)|\log a|+a\right)+e^{-\frac{\lambda}{2} N^2}\right].
$$
Therefore, in case $\delta^N(t)< 1,$ we choose $a =e^{-\lambda N^2}$,  yielding
\begin{equation}
\mathcal{G}_1+\mathcal{G}_2\leq C\left[N^{3\gamma'+2} \delta^N(t)+e^{-\frac{\lambda}{2} N^2}\right].\label{fdue}
\end{equation}
In case that $\delta^N(t)\geq1$ we come back to (\ref{g'})  (valid for any $\delta^N(t)$)
and we insert
the bound $\delta^N(t)\leq C N$, obtaining
$$
\mathcal{G}_1+\mathcal{G}_2\leq C\left[N^{3\gamma'} \log N \,  \delta^N(t)+e^{-\frac{\lambda}{2} N^2}\right],
$$
which can be included in (\ref{fdue}).

We come now to the  estimate of the term $\mathcal{G}_3.$ We have
\begin{equation}
\begin{split}
\mathcal{G}_3(x,v,t) \leq &\ |V^N(t)|| B(X^{N}(t))-B(X^{N+1}(t)|+\\&
\ |V^N(t)-V^{N+1}(t)|| B(X^{N+1}(t))|. 
\end{split}\end{equation}
By applying the Lagrange theorem to the magnetic field (as defined in (\ref{campo_B})) we have
\begin{equation*}
| B(X^{N}(t))-B(X^{N+1}(t)|\leq C\frac{|X^{N}(t)-X^{N+1}(t)|}{\left(\xi_1^N(t)\right)^{\tau+1}}
\end{equation*}
where $\xi_1^N(t)$ is the first coordinate  of a point   $\xi^N(t)$  of the segment joining $X^N(t)$ and $X^{N+1}(t).$ Due to the bound (\ref{B}), it has to be $\xi_1^N(t) \geq \frac{C}{N^{1/(\tau-1)}}.$
 Hence
 \begin{equation}
 | B(X^{N}(t))-B(X^{N+1}(t)|\leq CN^{\frac{\tau+1}{\tau-1}}\delta^N(t).
\end{equation}
This, together with the bounds (\ref{V^N}) and (\ref{B}), imply
\begin{equation}
\begin{split}
\mathcal{G}_3(x,v,t)& \leq \,C\left[N^{\frac{2\tau}{\tau-1 }} \delta^N(t)+N^{\frac{\tau}{\tau-1}}\eta^N(t)\right].\\&
 \label{F3}
\end{split}
\end{equation}
Finally, we come to the term $\mathcal{G}_4(x,v,t)$.  We have, by the form of the field $U$ in (\ref{campo_W}),
\begin{equation}
\left| \nabla U\left(X^N(t) \right) 
-\nabla U\left(X^{N+1}(t) \right) \right| \leq  C \frac{|X^{N}(t)-X^{N+1}(t)|}{\left( \zeta_1^N(t)\right)^{\mu +2}}
\end{equation}
where $\zeta_1^N(t)$ is the first coordinate  of a point   $\zeta^N(t)$  of the segment joining $X^N(t)$ and $X^{N+1}(t).$
Due to the bound (\ref{nablaW}), it has to be $\zeta_1^N(t) \geq \frac{C}{N^{1/(\tau-1)}}.$
 Hence
 \begin{equation}
\mathcal{G}_4(x,v,t)=
\left| \nabla U\left(X^N(t) \right) 
-\nabla U\left(X^{N+1}(t) \right) \right| \leq C N^{\frac{\mu+2}{\tau-1}}  \delta^N(t).
\label{G4}
 \end{equation}

At this point, going back to (\ref{iter_}), taking the supremum over the set $\{(x,v)\in \Lambda\times b(N)\},$ by (\ref{fdue}),  (\ref{F3}), and  (\ref{G4}) we arrive at:
\begin{equation}
\begin{split}
\delta^N(t) \leq \,&C\left(N^{3\gamma'+2} +N^{\frac{2\tau}{\tau -1}} + N^{\frac{\mu+2}{\tau-1}} \right) \int_0^t dt_1 \int_0^{t_1} dt_2\,\delta^N(t_2)+\\ &C \int_0^t dt_1 \int_0^{t_1} dt_2\,e ^{-\frac{\lambda}{2} N^2}+
CN^{\frac{\tau}{\tau-1}}\int_0^t dt_1 \int_0^{t_1} dt_2 \ \eta^N(t_2),\label{dd}
\end{split}
\end{equation}
where in (\ref{f1'}) we have taken into account  the bound (\ref{V^N}), which gives $|\delta^N(t)|\leq CN.$
Analogously, by using the same method to estimate the quantity $\eta^N(t),$ we get
\begin{equation}
\begin{split}
\eta^N(t) \leq \,&C\left(N^{3\gamma'+2} +N^{\frac{2\tau}{\tau -1}} + N^{\frac{\mu+2}{\tau-1}}\right) \int_0^t dt_1 \,\delta^N(t_1)+\\ &C \int_0^t dt_1 \,e ^{-\frac{\lambda}{2} N^2}+
CN^{\frac{\tau}{\tau-1}}\int_0^t dt_1  \, \eta^N(t_1).
\end{split}
\end{equation}
Since obviously $\delta^N(t_1)\leq \int_0^{t_1}dt_2\,\eta^N(t_2)$ we get from the last eqn.
\begin{equation}
\begin{split}
&\eta^N(t) \leq C\left(N^{3\gamma'+2} +N^{\frac{2\tau}{\tau -1}} + N^{\frac{\mu+2}{\tau-1}}\right) \int_0^t dt_1 \int_0^{t_1} dt_2\, \eta^N(t_2) \\
&+C \int_0^t dt_1\,e ^{-\frac{\lambda}{2} N^2}+CN^{\frac{\tau}{\tau-1}}\int_0^t dt_1  \ \eta^N(t_1).\label{ee}
\end{split}
\end{equation}

Hence, summing up (\ref{dd}) and (\ref{ee}), we have:

\begin{equation}
\begin{split}
\sigma^N(t)\leq \,&C\left(N^{3\gamma'+2} +N^{\frac{2\tau}{\tau -1}} + N^{\frac{\mu+2}{\tau-1}}\right) \int_0^t dt_1 \int_0^{t_1} dt_2\,\sigma^N(t_2)\\&+ CN^{\frac{\tau}{\tau-1}}\int_0^t dt_1  \ \sigma^N(t_1)+C \left(t+\frac{t^2}{2}\right)e ^{-\frac{\lambda}{2} N^2}.
\end{split}\label{d+e}
\end{equation}
We take $\nu$ such that
\begin{equation}
\max \left\{  3\gamma'+2, \frac{2\tau}{\tau-1}, \frac{\mu+2}{\tau-1} \right\} < 2\nu < 4  
\label{nu}  
\end{equation}
observing that such relation is compatible with the bounds for the parameters $\gamma'$, $\tau$, $\mu$, expressed
in Corollary \ref{coro}, $\gamma'<2/3$ and  $\tau > \frac74 \mu +\frac{11}{4}$.
Hence 
we get  
\begin{equation}
\begin{split}
\sigma^N(t)\leq \,C&\Bigg[N^{2\nu} \int_0^t dt_1 \int_0^{t_1} dt_2\,\sigma^N(t_2)+N^{\nu}\int_0^t dt_1  \ \sigma^N(t_1)+\\& \left(t+\frac{t^2}{2}\right)e ^{-\frac{\lambda}{2} N^2}\Bigg].
\end{split}\label{d'}
\end{equation}

We insert in the integrals the same inequality  for $\sigma^N(t_1)$ and $\sigma^N(t_2)$ and iterate in time, up to $k$ iterations. By direct inspection, using in the last step the estimate $\sup_{t\in [0,T]}\sigma^N(t)\leq CN,$ we arrive at
\begin{equation}
\begin{split}
\sigma^N(t)\leq &\,CNe ^{-\frac{\lambda}{2} N^2}\sum_{i=1}^{k-1}C^i\sum_{j=0}^{i}\binom{i}{j}\frac{N^{2\nu j}t^{2j}N^{\nu (i-j)}t^{i-j}}{(2j+i-j)!} +\\&C^kN \sum_{j=0}^{k}\binom{k}{j}\frac{N^{2\nu j}t^{2j}N^{\nu (k-j)}t^{k-j}}{(2j+k-j)!}.\end{split}
\label{double_sum}
\end{equation}
 The double sum in (\ref{double_sum}) collects all combinations of $j$ double time integrations of
 $(t+t^2/2)e^{-\frac{\lambda}{2}N^2}$ in (\ref{d'}), yielding the factor  $N^{2\nu j} t^{2 j}$ and the contribution $2j$ in the factorial at
 denominator, and of ($i-j$) single time integrations of $(t+t^2/2)e^{-\frac{\lambda}{2}N^2}$, yielding the factor $N^{\nu (i-j)} t^{i-j}$
 and the contribution ($i-j$) in the factorial at denominator; 
 the last single sum in (\ref{double_sum}) has the same structure, coming from terms of the iteration which avoid integration of $(t+t^2/2)e^{-\frac{\lambda}{2}N^2}$.  Note that  in absence of the middle term (single time integration) in
 the right hand side of (\ref{d'}) we would obtain
 $$
 \sigma^N(t) \leq  N C^k N^{2\nu k} \frac{t^{2k}}{(2k)!} + N e^{-\frac{\lambda}{2}N^2}
 \sum_{i=1}^{k-1} C^i N^{2\nu i} \frac{t^{2 i}}{(2 i)!}
 $$
 as in the case of  \cite{CCM18} without external fields.

By putting
\begin{equation*}
S_k''=\sum_{i=1}^{k-1}C^i\sum_{j=0}^{i}\binom{i}{j}\frac{N^{\nu (i+j)}t^{i+j}}{(i+j)!}
\end{equation*}
and 
\begin{equation*}
S_k'=C^k\sum_{j=0}^{k}\binom{k}{j}\frac{N^{\nu (j+k)}t^{j+k}}{(j+k)!}
\end{equation*}
we get
\begin{equation}
\sigma^N(t)\leq \,CNe ^{-\frac{\lambda}{2} N^2}S''_k+CNS_k'.\label{summ}
\end{equation}
 
 We start by estimating $S''_k.$ Recalling that $\binom{i}{j}<2^i$ we get
\begin{equation}
\begin{split}
S''_k\leq &\,\sum_{i=1}^{k-1}2^i C^i\sum_{j=0}^{i}\frac{N^{\nu (i+j)}t^{i+j}}{(i+j)!}.  
\end{split}
\end{equation}
The use of the Stirling formula $a^nn^n\leq n!\leq b^nn^n$ for some constants $a,b>0$ yields:
\begin{equation}
S''_k\leq \,\sum_{i=1}^{k-1}2^i\sum_{j=0}^{i}\frac{N^{\nu(i+j)}(Ct)^{i+j}}{(i+j)^{i+j}}\leq \sum_{i=1}^{k-1}2^i\frac{N^{\nu i}(Ct)^{i}}{i^{i}}
\sum_{j=0}^{i}\frac{N^{\nu j}(Ct)^{j}}{j^{j}},
\end{equation}
from which it follows, again by the Stirling formula,
\begin{equation}
S''_k\leq \left(e^{N^{\nu}Ct}\right)^2\leq e^{N^{\nu}C}.\label{sum'}
\end{equation}

For the term $S_k',$ putting $j+k=\ell,$ we get
\begin{equation}
S_k'\leq 2^k C^k\sum_{\ell=k}^{2k}\frac{N^{\ell\nu}(Ct)^{\ell}}{\ell^\ell}\leq  C^kk\frac{N^{2k\nu}}{k^k}.
\end{equation}
%The definition of the parameter $\gamma=\frac{15}{7}-\frac37 \alpha$ and the range of the parameter
%$\alpha\in(\frac83, 3)$ given in Proposition \ref{prop_1} guarantee that $\nu <.....$ 
By choosing $k=N^\zeta$  with $\zeta>  2\nu$,  we have, for sufficiently large $N$,
\begin{equation}
 S_k'\leq C^kk\left(k^{\frac{2\nu}{\zeta}-1}\right)^k\leq C^{-N^\zeta} \label{ffinal}
 \end{equation}
with $C>1$.

\noindent Going back to (\ref{summ}), by (\ref{sum'}) and (\ref{ffinal}) we have seen that
\begin{equation}
\sigma^N(t)\leq CN\left[e^{-\frac{\lambda}{2} N^2}e^{N^{\nu}C}+C^{-N^
\zeta}\right],
\end{equation}
and since
$\nu <2$
we have proved estimate (\ref{fff}).
\bigskip

\bigskip

\section{The estimate of $E^N$:  proof of Proposition \ref{field} } 
\label{sec_estE}

The proof of Proposition \ref{field} follows the same lines of \cite{KRM16}, leading to  a  refined estimate of $E^N$ in terms of the maximal velocity. To this aim we 
  need to control the time average of $E^N$ over a suitable time interval. Setting
\begin{equation*}
\langle E^N \rangle_{\bar{\Delta}}
:= \frac{1}{\bar{\Delta}} \int_{t}^{t+\bar{\Delta}} |E^N(X(s), s)| \, ds
\end{equation*}
we have the following result:
\begin{proposition}
There exists a positive number $\bar{\Delta},$ depending on $N,$ such that:
\begin{equation}
\langle E^N \rangle_{\bar{\Delta}} \leq
C {\cal{V}}^N(T)^{\gamma}  \qquad \qquad  \gamma < \frac23.
\label{averna}
\end{equation}
for any $t\in [0,T]$ such that $t\leq T- \bar\Delta$.
\label{media}\end{proposition}

In the proof of Proposition \ref{field}, we need to introduce several {\textit{positive}} parameters,  which we list here, for the convenience of the reader:
\begin{equation}
\begin{split}
&\frac{\mu}{2(\tau-1)}<\eta <\frac23-\frac{\mu}{\tau-1}, \quad  \frac43<\beta<2-\eta
-\frac{\mu}{\tau-1},\\& 0<\delta < 2-\eta-\beta-\frac{\mu}{\tau-1}:=c, \quad   0<\xi <\frac14 -\frac38\frac{\mu}{\tau-1}
\end{split} \label{para}
\end{equation} 
and in the sequel all the estimates hold  for {\textit{any}} choice of the previous parameters in the specified ranges
(note that the above intervals are non empty for what stated in Corollary \ref{coro}, i.e. $\frac{\mu+1}{\tau-1}<\frac49$).

\medskip

From now on we will skip the index $N$ through the whole section. Moreover, we put for brevity ${\cal{V}}:={\cal{V}}(T)$.

\begin{proof}

Let us define a time interval 
\begin{equation}
\Delta_1:=
\frac{1}{4 C_6{\cal{V}}^{\frac43+\frac{\mu}{3(\tau-1)}+\eta}}    \label{d1}
\end{equation}
where $C_6$ is the constant in (\ref{C1}). We  have to prove an estimate on $E$ analogous to  that in  \cite{KRM16}.
 For a positive integer $\ell$ we set:
\begin{equation}
\Delta_\ell= \Delta_{\ell-1}\mathcal{G} = ...=\Delta_1  \mathcal{G}^{\ell-1}, \label{dl}
\end{equation}
where 
\begin{equation}
\mathcal{G}= {\textnormal{Intg}}\left({\cal{V}}^{\delta} \right),
\label{effe_t}
\end{equation}
where ${\textnormal{Intg}}(a)$ is the integer part of $a.$ 
The integer part in (\ref{effe_t}) is taken in order to use  well known
properties of the average over intervals which increase by an integer factor (see Remark \ref{rem_av} in Section \ref{lemmas_proved}).

$Assume$ that the following estimate holds (it will be established in the next subsection \ref{proof_avern}), for any positive integer $\ell:$
\begin{equation}
\langle E \rangle_{\Delta_\ell} \leq
C \left[  {\cal{V}}^\gamma   +
   \frac{{\cal{V}}^{\frac43+\frac{\mu}{3(\tau-1)} -c}}{ {\cal{V}}^{\delta (\ell -1)}}  \right].
\label{avern}
\end{equation}
 Hence, defining $\bar\ell$ as the smallest integer such that
\begin{equation}
\delta (\bar\ell -1) > \frac23+\frac{\mu}{3(\tau-1)} -c,     
\label{Lt}
\end{equation}
estimate (\ref{avern}) implies (\ref{averna}) with $\bar{\Delta}:=\Delta_{\bar\ell}$, and this concludes
the proof of Proposition \ref{media}.
\end{proof}

\noindent It can be seen that 
\begin{equation*}
\frac{C}{\mathcal{V}^{\frac23 +\eta+c}}\leq\bar{\Delta} 
\leq \frac{C}{\mathcal{V}^{\frac23 +\eta}}.     
\end{equation*}
\medskip

\noindent We observe that Proposition \ref{field} follows immediatly from Proposition \ref{media},
by dividing the time integration interval in $n$ subintervals such that we can apply Proposition \ref{media} on each of them.
Then  in order to prove Proposition \ref{field}, we have to prove that (\ref{avern}) holds, which we will do in the next subsection. 

\subsection{Proof of (\ref{avern})}
\label{proof_avern}

\medskip

To prove (\ref{avern}) we need some preliminary results, which are stated here and proved in Section \ref{lemmas_proved}.

 We consider two solutions of the partial dynamics, $\left(X(t),V(t)\right)$ and $\left(Y(t),W(t)\right),$ starting from $(x,v)$ and $(y,w)$ respectively. Let $\eta$ be the parameter introduced in the definition (\ref{d1}) of $\Delta_1.$ Then we have:

\begin{lemma}\label{lemv3}
 Let $t\in [0,T]$ such that $t+\Delta_\ell\in [0,T]$.
$$
\hbox{If} \qquad |V_3(t)-W_3(t)|\leq {\cal{V}}^{-\eta} $$
then 
\begin{equation}
\sup_{s\in [t, t+\Delta_\ell]}|V_3(s)-W_3(s)|\leq 2{\cal{V}}^{-\eta}. \label{L1}
\end{equation}

\medskip

$$
\hbox{If} \qquad|V_3(t)-W_3(t)|\geq {\cal{V}}^{-\eta}
$$
then 
\begin{equation}
\inf_{s\in [t, t+\Delta_\ell]}|V_3(s)-W_3(s)|\geq \frac12 {\cal{V}}^{-\eta}. \label{L2} 
\end{equation}
\end{lemma}

\medskip
\noindent For a generic vector $v$ we put $v_{\perp}=(v_1,v_2, 0).$
\begin{lemma}\label{lemperp}
 Let $t\in [0,T]$ such that $t+\Delta_\ell \in [0,T]$.  
 $$
\hbox{If} \quad |V_\perp(t)|\leq {\cal{V}}^{\xi} 
$$
then
\begin{equation}
\sup_{s\in [t, t+\Delta_\ell]}|V_\perp(s)| \leq 2{\cal{V}}^{\xi} .   \label{L3}
\end{equation}

\medskip

\noindent 
$$
\hbox{If} \quad |V_\perp(t)|\geq {\cal{V}}^{\xi}   
$$
then
\begin{equation}
\inf_{s\in [t, t+\Delta_\ell]}|V_\perp(s)|\geq \frac{ {\cal{V}}^{\xi} }{2}.    \label{L4}
\end{equation}
\end{lemma}
\begin{lemma}\label{lemrect}
 Let $t\in [0,T]$ such that  $t+\Delta_\ell\in [0,T]$       
  and assume that $|V_3(t)-W_3(t)|\geq h {\cal{V}}^{-\eta}$ for some $h\geq 1$. 
Then it exists $t_0\in [t, t+\Delta_\ell]$ such that for any $s\in [t, t+\Delta_\ell] $ it holds:
$$
|X(s)-Y(s)|\geq \frac{h {\cal{V}}^{-\eta}}{4}|s-t_0|.
$$
\end{lemma}

Estimate (\ref{avern}) will be proved  analogously to what has been done in \cite{KRM16}. 
We use an inductive procedure, that is:

\noindent
$step\ i)$ we prove (\ref{avern}) for $\ell=1;$ 

\noindent
$step\ ii)$ we show that if (\ref{avern}) holds for $\ell-1$ it holds also for $\ell.$  

Step $i)$ is the fundamental one, while step $ii)$ is an almost immediate consequence, as it will be seen after.

\medskip
\noindent
\textbf{Proof  of step \textit{i}).}

All the parameters appearing in this demonstration have been listed in (\ref{para}). 

We have to show that the following estimate holds:
\begin{equation}
\langle E \rangle_{\Delta_{1}} \leq
C \left[  {\cal{V}}^\gamma +
   {\cal{V}}^{\frac43+\frac{\mu}{3(\tau-1)}-c}   \right].
\label{avern1}
\end{equation}
We fix any $t\in [0,T]$ such that $t+\Delta_1\leq T$,  and we consider the time evolution of the characteristics over the time interval 
$[t,  t+\Delta_1]$.  For any $s\in [t, t+\Delta_1]$ we set 
\begin{equation*}
\begin{split}
&(X(s),V(s)):=(X(s,t,x,v),V(s,t,x,v)), \quad X(t)=x\\& (Y(s),W(s)):=(Y(s,t,y,w),W(s,t,y,w)), \quad Y(t)=y.
\end{split}\end{equation*}

Then
\begin{equation}
\begin{split}
|E(X(s),s)|\leq & \int d\xi d\omega \,  \frac{f(\xi,\omega,s)}{|X(s)-\xi|^2}=
\int dydw \ \frac{ f(y,w,t)}{|X(s)-Y(s)|^2}.\label{Ei}
\end{split}
\end{equation}

We decompose the phase space $(y,w)\in {\mathbb{R}}^3\times {\mathbb{R}}^3$ in the following way. We define

\begin{equation}
T_1=\{(y,w): y_1>0, \,\, |y-x|\leq 2R(T)\} 
\label{T1}
\end{equation}
\begin{equation}
S_1= \{ (y,w): |v_3-w_3|\leq {\cal{V}}^{-\eta} \}\label{S1}
\end{equation}
\begin{equation}
S_2=\{ (y,w): \ |w_{\perp}|\leq {\cal{V}}^{\xi} \}\label{S2}
\end{equation}
\begin{equation}
S_3=\{ (y,w): \  |v_3-w_3|> {\cal{V}}^{-\eta} \}\cap  \{ (y,w): |w_{\perp}|> {\cal{V}}^{\xi} \}.
\end{equation}
 We have
\begin{equation}
|E(X(s),s)|\leq\sum_{j=1}^4{\mathcal{I}}_j(X(s))\label{sum}
\end{equation}
where for any $s\in [t, t+\Delta_1]$
\begin{equation*}
{\mathcal{I}}_j(X(s))=\int 
_{T_1\cap S_j}dydw \  \frac{f(y,w,t)}{|X(s)-Y(s)|^2}, \quad \quad  j=1,2,3 
\end{equation*}
and 
\begin{equation*}
{\mathcal{I}}_4(X(s))=\int 
_{T_1^c}dydw \  \frac{f(y,w,t)}{|X(s)-Y(s)|^2}.
\end{equation*}
Let us start by the first integral. By the change of variables $(Y(s),W(s))=(\bar{y},\bar{w})$,
and Lemma \ref{lemv3} we get
\begin{equation}
{\mathcal{I}}_1(X(s))\leq \int _{T_1'\cap S_1'}d\bar{y}d\bar{w}  \frac{f(\bar{y},\bar{w},s)}{|X(s)-\bar{y}|^2}
\end{equation}
where $T_1'=\{(\bar{y}, \bar{w}):  \bar{y}_1>0, \,\, |\bar{y}-X(s)|\leq 4R(T)\}$ and $S_1'= \{ (\bar{y},\bar{w}): |V_3(s)-\bar{w}_3|\leq 2 {\cal{V}}^{-\eta} \}.$ 
 Now it is:
 \begin{equation}
\begin{split}
{\mathcal{I}}_1(X(s))\leq &\int _{T_1'\cap S_1'\cap \{|X(s)-\bar{y}|\leq \varepsilon\}}d\bar{y}d\bar{w} \frac{\  f(\bar{y},\bar{w},s)}{|X(s)-\bar{y}|^2}+\\&
\int _{T_1'\cap S_1'\cap \{  |X(s)-\bar{y}|>\varepsilon\}}d\bar{y}d\bar{w}\ \frac{f(\bar{y},\bar{w},s)}{|X(s)-\bar{y}|^2}\leq \\& C{\cal{V}}^{2-\eta} \varepsilon+ \int _{T_1'\cap S_1'\cap \{  |X(s)-\bar{y}|>\varepsilon\}}d\bar{y}d\bar{w}\ \frac{f(\bar{y},\bar{w},s)}{|X(s)-\bar{y}|^2}.\end{split}
\label{i11}\end{equation}
Now we give a bound on the spatial density $\rho(\bar{y},s).$ Setting 
$$
\rho_1(y,s)=\int_{S_1'} dw f(y,w,s),
$$
we have:
\begin{equation}
\begin{split}
&\rho_1(y,s)\leq  C\, {\cal{V}}^{-\eta} \int_{|w_{\perp}|\leq a} dw_{\perp}+ \nonumber 
 \int_{|w_{\perp}|> a}dw_{\perp}\int dw_3 \, f(y,w,s) \leq   \nonumber \\
& C  a^2 {\cal{V}}^{-\eta}+\frac{1}{a^2}\int dw |w|^2f(y,w,s)=C  a^2 {\cal{V}}^{-\eta}+\frac{1}{a^2}K(y,s)
\nonumber
\end{split}
\end{equation}
where $K(y,s)=\int dw |w|^2f(y,w,s).$ Minimizing in $a$ we obtain
\begin{equation}
\rho_1(y,s)\leq C {\cal{V}}^{-\frac{\eta}{2}} K(y,s)^{\frac12}.\label{K}
\end{equation}

Hence from (\ref{K}) and (\ref{kinet}) we get
\begin{equation}
\begin{split}
&\Bigg(\int _{T_1'}dy\ \rho_1(y,s)^2\Bigg)^{\frac12}\leq \, C {\cal{V}}^{-\frac{\eta}{2}} \Bigg(\int_{T_1'} dy\ K(y,s)\Bigg)^{\frac12}  \\
& \leq  C {\cal{V}}^{-\frac{\eta}{2}} \left(   {\mathcal{V}}^{\frac{\mu}{\tau -1}}  \right)^{\frac12}
=C {\cal{V}}^{-\frac{\eta}{2}+\frac{\mu}{2(\tau -1)}}. 
\end{split}
\label{rho1}
\end{equation}
 Going back to (\ref{i11}), this bound implies
\begin{equation*}
\begin{split}
&{\mathcal{I}}_1(X(s))\leq\\&  C{\cal{V}}^{2-\eta} \varepsilon+\Bigg(\int_{T_1'} d\bar{y} \,\rho_1(\bar{y},s)^2\Bigg)^{\frac12}\Bigg(\int_{T_1'\cap  \{|X(s)-\bar{y}|>\varepsilon\}} d\bar{y} \  \frac{1}{|X(s)-\bar{y}|^4}\Bigg)^{\frac12} \\&
\leq C\left({\cal{V}}^{2-\eta}  \varepsilon+\sqrt{\frac{{\cal{V}}^{-\eta+\frac{\mu}{\tau-1}} }{\varepsilon}}\right).
\end{split}
\end{equation*}
Minimizing in $\varepsilon$ we obtain:
\begin{equation}
{\mathcal{I}}_1(X(s))\leq C {\cal{V}}^{\frac23-\frac{2\eta}{3}+\frac{\mu}{3(\tau-1)}}  \label{I1}
\end{equation}
 then the lower bound for $\eta$ in (\ref{para})  implies that ${\mathcal{I}}_1(X(s))$
is bounded by a power of ${\cal{V}}$ less than $\frac23$.

Now we pass to the term ${\mathcal{I}}_2.$ Proceeding as for the term ${\mathcal{I}}_1,$ defining $S_2'=\{ (y,w): \ |w_{\perp}|\leq 2 {\cal{V}}^{\xi} \},$  by Lemma \ref{lemperp} and the H\"older's inequality we get:
$$
{\mathcal{I}}_2(X(s))\leq  \int _{T_1'\cap S_2'}d\bar{y}d\bar{w} \frac{\  f(\bar{y},\bar{w},s)}{|X(s)-\bar{y}|^2}\leq 
$$
$$
 \int _{T_1'\cap S_2'\cap \{|X(s)-\bar y|\leq \varepsilon\}}d\bar{y}d\bar{w} \frac{\  f(\bar{y},\bar{w},s)}{|X(s)-\bar{y}|^2}+
 \int _{T_1'\cap \{  |X(s)-\bar{y}|>\varepsilon\}}d\bar{y} \frac{  \rho(\bar{y},s)}{|X(s)-\bar{y}|^2}\leq $$
$$
C {\cal{V}}^{1+2\xi}\varepsilon+
\left(\int_{T_1'}d\bar{y} \ \rho(\bar{y},s)^{\frac53}\right)^{\frac35}\left(\int_{ \{|X(s)-\bar{y}|> \varepsilon\}} d\bar{y} \  \frac{1}{|X(s)-\bar{y}|^5}\right)^{\frac25}.
$$
The bound (\ref{lem1}) implies
$$
{\mathcal{I}}_2(X(s))\leq C{\cal{V}}^{1+2\xi}\varepsilon+C {\cal{V}}^{\frac{3\mu}{5(\tau-1)}} \varepsilon^{-\frac45}.
$$
By minimizing in $\varepsilon$ we get:
\begin{equation}
{\mathcal{I}}_2(X(s))\leq C {\cal{V}}^{\frac49+\frac89 \xi +\frac{\mu}{3(\tau-1)}}.    
\label{I2}
\end{equation}
Analogously to what we have seen before for the term ${\mathcal{I}}_1(X(s)),$ also in this case the upper bound for the parameter $\xi$ in (\ref{para}) implies that ${\mathcal{I}}_2(X(s))$
is bounded by a power of ${\cal{V}}$ less than $\frac23.$

 Now we estimate ${\mathcal{I}}_3(X(s)).$ It will be clear in the sequel that it is because of this term that we are forced to bound the time average of $E,$ and then to iterate the bound, from smaller to larger time intervals. 
 
 We cover $ T_1\cap S_3$ by means of the sets
$A_{h,k}$ and  $B_{h,k}$, with 
${k=0,1,2,...,m}$ and ${h=1,2,...,m'}$, defined in the following way:
\begin{equation}
\begin{split}
A_{h,k}=\big\{ &(y,w): \ h {\cal{V}}^{-\eta}< |v_3-w_3|\leq (h+1) {\cal{V}}^{-\eta},\\& \alpha_{k+1}< |w_{\perp}|\leq \alpha_k, \  |X(s)-Y(s)|\leq l_{h,k}\big\}
\end{split}\label{Ak}
\end{equation}
\begin{equation}
\begin{split} 
B_{h,k}=\big\{ &(y,w):\ h {\cal{V}}^{-\eta}< |v_3-w_3|\leq (h+1) {\cal{V}}^{-\eta},\\& \alpha_{k+1}< |w_{\perp}|\leq \alpha_k, \  |X(s)-Y(s)|> l_{h,k}\big\}
\end{split}\label{Bk}
\end{equation}
where:
\begin{equation}
\alpha_k=\frac{{\cal{V}}}{2^k} \quad \quad l_{h,k}=\frac{2^{2k} }{h{\cal{V}}^{\beta-\eta} },
\label{al}
\end{equation}
with $\beta$ chosen in (\ref{para}).
Consequently we put
\begin{equation}
{\cal{I}}_3(X(s))\leq\sum_{h=1}^{m'} \sum_{k=0}^m\left({\cal{I}}_3'(h,k)+{\cal{I}}_3''(h,k)\right) \label{23}
\end{equation}
being
\begin{equation}
{\cal{I}}_3'(h,k)=\int_{T_1\cap A_{h,k}} \frac{f(y,w,t)}{|X(s)-Y(s)|^2} \, dydw
\end{equation}
and
\begin{equation}
{\cal{I}}_3''(h,k)=\int_{T_1\cap B_{h,k}} \frac{f(y,w,t)}{|X(s)-Y(s)|^2} \, dydw.
\end{equation}
Since we are in $S_3,$ it is immediately seen that 
\begin{equation}
m\leq C \log {\cal{V}},  \qquad m'\leq C {\cal{V}}^{1+\eta}.\label{par}
\end{equation} 
By adapting Lemma \ref{lemv3} and Lemma \ref{lemperp} to this context it is easily seen that $\forall \ (y,w)\in A_{h,k}$ it holds:
\begin{equation}
(h-1) {\cal{V}}^{-\eta}\leq |V_3(s)-W_3(s)|\leq (h+2) {\cal{V}}^{-\eta},   \label{lem31}
\end{equation}
and
\begin{equation}
\frac{\alpha_{k+1}}{2}\leq |W_{\perp}(s)|\leq 2\alpha_k.\label{lemperp1}
\end{equation}
Hence, setting
\begin{equation}\begin{split}
A_{h,k}'=&\big\{(\bar{y},\bar{w}):  \ (h-1){\cal{V}}^{-\eta}\leq|V_3(s)-\bar{w}_3|\leq(h+2){\cal{V}}^{-\eta},  \\&
\frac{\alpha_{k+1}}{2}\leq |\bar{w}_{\perp}|\leq 2\alpha_k,  
\ |X(s)-\bar{y}|\leq l_{h,k}  \big\}. 
\end{split}
\label{Akk}
\end{equation}
we have
\begin{equation}
{\cal{I}}_3' (h,k)\leq
\int_{T_1'\cap A_{h,k}'}\frac{f(\bar{y},\bar{w}, s)}{|X(s)-\bar{y}|^2} \, d\bar{y}d\bar{w}\label{int3}.
\end{equation}
By the choice of the parameters $\alpha_k$ and $l_{h,k}$ made in (\ref{al}) we have:
\begin{equation}
\begin{split}
{\cal{I}}_3' (h,k) \leq &
 \ C \, l_{h,k}\int_{A_{h,k}'} \, d\bar{w}  \leq C l_{h,k}\alpha^2_{k}\int_{A_{h,k}'}d\bar{w}_3\leq \\&
 \,\, C \, l_{h,k} \alpha_{k}^2  {\cal{V}}^{-\eta} 
\leq C\frac{{\cal{V}}^{2-\beta} }{h }. 
\end{split}
\end{equation}

Hence by (\ref{par})
\begin{equation}
\sum_{h=1}^{m'} \sum_{k=0}^m{\cal{I}}_3'(h,k)\leq C {\cal{V}}^{2-\beta}  \sum_{k=0}^m \sum_{h=1}^{m'}\frac{1}{h}\leq 
C {\cal{V}}^{2-\beta} \log^2{\cal{V}}.
\label{i3}\end{equation}
The lower bound for $\beta$ in (\ref{para}) is such that 
\begin{equation*}
\sum_{h=1}^{m'} \sum_{k=0}^m{\cal{I}}_3'(h,k)\leq C {\cal{V}}^\gamma \qquad \gamma <\frac23.
\end{equation*}

Now we pass to ${\mathcal{I}}_3''(h,k).$
  Setting
\begin{equation}
\begin{split}
C_{h,k}=\big\{& (y,w): \, h {\cal{V}}^{-\eta}\leq |v_3-w_3|\leq (h+1) {\cal{V}}^{-\eta},\\& \alpha_{k+1}\leq |w_{\perp}|\leq \alpha_k\big\},
\end{split}
\label{Bhk}
\end{equation}
 we have:
\begin{equation}
\begin{split}
\int_{t}^{t+\Delta_1}& {\mathcal{I}}_3''(h,k)\ ds\leq   \\
&\int_{t}^{t+\Delta_1}ds\int_{T_1\cap C_{h,k}}dydw \, \frac{ f(y,w,t)}{|X(s)-Y(s)|^2} \chi(|X(s)-Y(s)|> l_{h,k})  \leq
\\&\int_{T_1\cap C_{h,k}} f(y, w, t) \left( \int_{t}^{t+\Delta_1} \frac{\chi(|X(s)-Y(s)|> l_{h,k})}{|X(s)-Y(s)|^2} \, ds \right)\, dy dw.
\end{split}
\label{ik}
\end{equation}
By Lemma \ref{lemrect},
putting $a = \frac{4 \, l_{h,k}{\cal{V}}^{\eta}}{h }$  we have:
\begin{equation}
\begin{split}
&\int_{t}^{t+\Delta_1} \frac{\chi(|X(s)-Y(s)|> l_{h,k})}{|X(s)-Y(s)|^2} \, ds 
\leq  \,\,   \\
&\int_{\{ s: |s-t_0|\leq a \}} \frac{\chi(|X(s)-Y(s)|> l_{h,k})}{|X(s)-Y(s)|^2} \, ds \, + \\
&\int_{\{ s: |s-t_0| > a \}} \frac{\chi(|X(s)-Y(s)|> l_{h,k})}{|X(s)-Y(s)|^2} \, ds   \leq  \\
&\frac{1}{l_{h,k}^2} \int_{\{ s: |s-t_0|\leq a \}} \, ds
+\frac{C\, {\cal{V}}^{2\eta}}{h^2} \int_{\{ s: |s-t_0| > a \}} \frac{1}{ |s-t_0|^2} \, ds \leq \\
& \quad \frac{2 a}{l_{h,k}^2} + \frac{C\,{\cal{V}}^{2\eta}}{h^2 } \int_a^{+\infty} \frac{1}{s^2} \, ds
= \frac{C\, {\cal{V}}^{\eta}}{ l_{h,k}h }.
\end{split}
\label{eq1}\end{equation}
Moreover:
\begin{equation}
\begin{split}
\int_{T_1\cap C_{h,k}} f(y, w, t)\, dydw&\leq \frac{C}{\alpha_k^2}\int_{T_1\cap C_{h,k}} w^2 f(y, w, t) \, dydw \\& \end{split}
\label{eq2}
\end{equation}
so that, by (\ref{kinet}),
\begin{equation}
\sum_{h=1}^{m'} \sum_{k=0}^m \int_{T_1\cap C_{h,k}} w^2 f(y, w,t) \, dydw
\leq C \int_{T_1} K(y,t) \, dy \leq C {\cal{V}}^{\frac{\mu}{\tau-1}}. 
\label{i5}
\end{equation}

Taking into account (\ref{al}), by (\ref{ik}), (\ref{eq1}), (\ref{eq2}) and (\ref{i5}) we get:
\begin{equation}
\sum_{h=1}^{m'} \sum_{k=0}^m\int_{t}^{t+\Delta_1} {\mathcal{I}}_3''(h,k)\, ds\leq  
C  {\cal{V}}^{\beta-2+\frac{\mu}{\tau-1}}. 
\label{wer}
\end{equation}
By multiplying and dividing by $\Delta_{1}$ defined in (\ref{d1}) we obtain, 
\begin{equation}
\sum_{h=1}^{m'} \sum_{k=0}^m\int_{t}^{t+\Delta_1} {\mathcal{I}}_3''(h,k)\, ds\leq \, 
C {\cal{V}}^{\frac43+\frac{\mu}{3(\tau-1)}}  {\cal{V}}^{-2+\eta+\beta+\frac{\mu}{\tau-1}}  \, \Delta_{1}. \label{i3''}
\end{equation}
Now we have %by Corollary \ref{cor},
\begin{equation}
 {\cal{V}}^{-2+\eta+\beta+\frac{\mu}{\tau-1}} 
:=  {\cal{V}}^{-c},
\label{def_c}
\end{equation}
hence, since $c$ is positive  (see (\ref{para})), the estimate (\ref{C1}) of the electric field in Proposition \ref{prop2} can be improved, at least on a short time interval $\Delta_1$, and this fact allows the inductive method to work.

\bigskip

Finally 
the bounds (\ref{I1}), (\ref{I2}), (\ref{23}), (\ref{i3}) and (\ref{i3''}) imply:
\begin{equation}
\begin{split}
&\sum_{j=1}^3\int_{t}^{t+\Delta_1} {\mathcal{I}}_j(X(s))\, ds\leq 
 C\Delta_{1}\left[{\cal{V}}^{\gamma}  +
{\cal{V}}^{\frac43+\frac{\mu}{3(\tau-1)}-c}  \right].
\end{split} \label{ave}
\end{equation}

It remains the estimate of the last term ${\mathcal{I}}_4(X(s)).$ It can be treated
as in \cite{JSP17}, obtaining
\begin{equation}
{\mathcal{I}}_4(X(s))\leq C.
\label{i6}
\end{equation}
Hence by (\ref{Ei}), (\ref{sum}) and (\ref{ave}), this last bound implies:
\begin{equation*}
 \int_{t}^{t+\Delta_1} |E(X(s),s)| \, ds\leq
C\, \Delta_{1}\left[{\cal{V}}^{\gamma}  +
{\cal{V}}^{\frac43+\frac{\mu}{3(\tau-1)}-c}  \right],
\end{equation*}
so that we have proved (\ref{avern}) for $\ell=1.$

\medskip

\noindent
\textbf{Proof  of step \textit{ii}).}

\noindent In the previous step we have seen that, starting from estimate (\ref{C1}) on $[0,T]$, we arrive at (\ref{avern1}) on $\Delta_1.$

\noindent We show now that, assuming true (\ref{avern}) over $\Delta_{\ell-1}$, then it holds also over $\Delta_{\ell}$,
obtaining so  an improved estimate.
  We note that (see Remark \ref{rem_av} in Section \ref{lemmas_proved})  if estimate (\ref{avern}) holds for $\langle E \rangle_{\Delta_{\ell-1}}$, then the same estimate holds for $\langle E \rangle_{\Delta_{\ell}},$   since both the intervals $\Delta_{\ell}$ and the bound (\ref{avern}) are uniform in time.

We recall that the term $\mathcal{I}_3$ was the only  one for which we needed to do the time average. Hence, proceeding in analogy to what we have done above,  by using the estimate on the time average 
$\langle E \rangle_{\Delta_{\ell-1}}$ on the larger time interval $\Delta_\ell$  we get for the term
(\ref{wer})
\begin{equation}
\begin{split}
\sum_h \sum_k \int_{t}^{t+\Delta_\ell} \mathcal{I}_3''(h,k)\, ds  \leq  C    {\cal{V}}^{\beta-2+\frac{\mu}{\tau-1}}
\frac{\Delta_\ell}{\Delta_\ell}
\leq 
C\frac{{\cal{V}}^{\frac43+\frac{\mu}{3(\tau-1)}-c}}{ {\cal{V}}^{(\ell-1)\delta}}  \,
\Delta_\ell
\end{split}
\end{equation}
and consequently
\begin{equation}
\langle E \rangle_{\Delta_\ell} \leq
C \left[  {\cal{V}}^{\gamma}  
 +  \frac{{\cal{V}}^{\frac43+\frac{\mu}{3(\tau-1)}-c}}{ {\cal{V}}^{(\ell-1)\delta}}  \right]
 \label{aver2}
\end{equation}
which proves the second step. Hence (\ref{avern}) is proved for any $\ell.$

\section{Proofs of Lemmas \ref{lemv3}, \ref{lemperp}, \ref{lemrect}}
\label{lemmas_proved}

Also in this section we will skip the index $N$ in the estimates, but it has to be reminded that the following estimates concern the partial dynamics.

\medskip

\begin{remark}
Before giving the proofs of the Lemmas  
 we observe that it holds:
\begin{equation}
\sup_{t\in[0, T-\Delta_\ell]}\langle E\rangle_{\Delta_\ell}\leq 
\sup_{t\in [0, T-\Delta_{\ell -1}]} \langle E\rangle_{\Delta_{\ell-1}} \qquad \forall \ell \leq \bar{\ell}.\label{ave}
\end{equation}
In fact, 
$\Delta_\ell=\mathcal{G}\Delta_{\ell-1}$ with $\mathcal{G}$ given in (\ref{effe_t}), so that:
\begin{equation*}
[t,t+\Delta_{\ell}]=\bigcup_{i=1}^\mathcal{G}[t+(i-1)\Delta_{\ell-1}, t+i\Delta_{\ell-1}],
\end{equation*}
hence
\begin{equation}
\frac{1}{\Delta_{\ell}}\int_t^{t+\Delta_{\ell}}|E(X(s),s)|ds \leq \max_i \frac{1}{\Delta_{\ell-1}}\int_{t+(i-1)\Delta_{\ell-1}}^{t+i\Delta_{\ell-1}}|E(X(s),s)|ds,
\end{equation}
whence we get (\ref{ave}).

\label{rem_av}
\end{remark}

\bigskip

\noindent \textbf{Proof of Lemma \ref{lemv3}}.

\medskip

We give first the proof for $\ell=1$, that is $\Delta_\ell=\Delta_1$.

\noindent Since the magnetic force and the external field $U$ give no contribution to the third component of the velocity, 
by (\ref{C1})  and (\ref{d1}) we get, for any $s\in[t, t+\Delta_1]$,
\begin{equation*}
\begin{split}
|V_3(s) - W_3(s)| &\leq |V_3(t)-W_3(t)| +\\& \int_{t}^{t+\Delta_1} \Big[ |E(X(s),s)| + |E(Y(s),s)| \Big] ds 
\leq \\& {\cal{V}}^{-\eta} + 2 C_6 {\cal{V}}^{\frac43+\frac{\mu}{3(\tau-1)}}  \Delta_1 \leq 2 {\cal{V}}^{-\eta}.
\end{split}
\end{equation*}

Analogously we prove the second statement:
\begin{equation*}
\begin{split}
|V_3(s)-W_3(s)|&\geq |V_3(t)-W_3(t)|-\\&\int_{t}^{t+\Delta_1}
\Big[ |E(X(s),s)| +|E(Y(s),s)| \Big] ds
\geq\\& {\cal{V}}^{-\eta} - 2 C_6 {\cal{V}}^{\frac43+\frac{\mu}{3(\tau-1)}}   \Delta_1
\geq \frac12 {\cal{V}}^{-\eta}.
\end{split}
\end{equation*}

\medskip

We show now that Lemma \ref{lemv3} holds true also over a time interval $\Delta_\ell$, $\ell > 1$, supposing for the electric field
the estimate (\ref{avern}) at level $\ell - 1$ (that is, only Lemma \ref{lemv3} at level less than $\ell$ is needed
to establish (\ref{avern}) at level $\ell$).  By Remark \ref{rem_av},   proceeding as before we get  for any $s\in[t, t+\Delta_\ell]$,  
\begin{equation}
\begin{split}
|&V_3(s) - W_3(s)| \leq |V_3(t)-W_3(t)| +\\& \int_{t}^{t+\Delta_\ell} \Big[ |E(X(s),s)| + |E(Y(s),s)| \Big] ds 
\leq \\&{\cal{V}}^{-\eta} +  C   
\left[  {\cal{V}}^{\gamma} +  
\frac{{\cal{V}}^{\frac43+\frac{\mu}{3(\tau-1)} - c} }{ {\cal{V}}^{\delta(\ell-2)}} 
  \right]  
  \frac{{\cal{V}}^{\delta(\ell-1)}}{4 C_6  {\cal{V}}^{\frac43+\frac{\mu}{3(\tau-1)}+\eta} } \leq \\ 
	&{\cal{V}}^{-\eta} +  C {\cal{V}}^{\gamma-\frac43-\frac{\mu}{3(\tau-1)}-\eta+\delta(\ell-1)}  
	+C {\cal{V}}^{\delta-\eta-c}    \leq \\
	&{\cal{V}}^{-\eta} + C {\cal{V}}^{-\eta-\varepsilon}   \leq 2 {\cal{V}}^{-\eta}, \end{split}
\label{vareps}
\end{equation}
(for a suitable small $\varepsilon>0$)
by  (\ref{Lt}) and the choice of the parameters made in (\ref{para}), considered that the constant
$C_3$ appearing in the definition of $\mathcal{V}$ (\ref{maxV}) is sufficiently large.

We proceed analogously for the lower bound. 

\bigskip

\noindent  \textbf{Proof of Lemma \ref{lemperp}}.

\medskip

We begin with the case $\ell=1$, that is $\Delta_\ell=\Delta_1$.

\noindent 
We prove the thesis by contradiction. Assume that there exists a time interval $[t^*,t^{**}]\subset [t, t+\Delta_1)$, 
such that $|V_{\perp}(t^*)|= {\cal{V}}^{\xi},$  $|V_{\perp}(t^{**})|= 2{\cal{V}}^{\xi}$  and  
${\cal{V}}^{\xi}< |V_{\perp}(s)|< 2{\cal{V}}^{\xi} \ \ \forall s\in (t^*,t^{**}).$ By the definition of $B$ it is:
\begin{equation}
\begin{split}
\frac{d}{dt}V_{\perp}^2(t)=2V_{\perp}(t)\cdot E(X(t),t)-2\nabla U(X(t))\cdot V_{\perp}(t), \label{perp}\end{split}
\end{equation}
so that, by  (\ref{nablaW'}) and (\ref{C1})  (recalling it is
$\frac{\mu+1}{\tau-1}<\frac49$)
\begin{equation}
\begin{split}
4{\cal{V}}^{2\xi}=|V_{\perp}(t^{**})|^2\leq \ |V_{\perp}(t^{*})|^2&+2\int_{t^*}^{t^{**}}  ds \, |V_{\perp}(s)|\,
|E(X(s),s)|  \\
& +2 \int_{t^*}^{t^{**}}  ds \, |V_{\perp}(s)|\,
|\nabla U(X(s))| \leq \\
& {\cal{V}}^{2\xi}+4{\cal{V}}^{\xi}\int_{t^*}^{t^{**}} ds \, | E(X(s),s)|  \\
& + C \, {\cal{V}}^{\xi + \frac{\mu+1}{\tau -1}} \Delta_1 \leq  \\
& {\cal{V}}^{2\xi}+4{\cal{V}}^{\xi} \Delta_1  C_6 {\cal{V}}^{\frac43+\frac{\mu}{3(\tau-1)}}  + C \, {\cal{V}}^{\xi + \frac{\mu+1}{\tau -1}} \Delta_1         < \\
&2 {\cal{V}}^{2\xi}.
\end{split}
\label{app1}
\end{equation}
Hence it is absurd to assume $|V_{\perp}(t^{**})|= 2{\cal{V}}^{\xi}$, on the contrary $|V_{\perp}|$
will remain smaller than $2{\cal{V}}^{\xi}$ in the whole interval $[t, t+\Delta_1]$, thus proving (\ref{L3}).

Now we prove (\ref{L4}).  As before, assume that there exists a time interval $[t^*,t^{**}]\subset [t, t+\Delta_1)$, 
such that 
$|V_{\perp}(t^*)|= {\cal{V}}^{\xi},$  $|V_{\perp}(t^{**})|= \frac12{\cal{V}}^{\xi}$  and  
$\frac12{\cal{V}}^{\xi}< |V_{\perp}(s)|< {\cal{V}}^{\xi}$   $\, \forall s\in (t^*,t^{**})$. Then from (\ref{perp})  it follows, by (\ref{nablaW'}) and (\ref{C1})  
\begin{equation}
\begin{split}
\frac14{\cal{V}}^{2\xi}=|V_{\perp}(t^{**})|^2\geq \ |V_{\perp}(t^{*})|^2&-2\int_{t^*}^{t^{**}}  ds \, |V_{\perp}(s)|\,
|E(X(s),s)|  \\
& -2 \int_{t^*}^{t^{**}}  ds \, |V_{\perp}(s)|\,
|\nabla U(X(s))| \geq \\
& {\cal{V}}^{2\xi}-2{\cal{V}}^{\xi}\int_{t^*}^{t^{**}} ds \, | E(X(s),s)|  \\
& - C \, {\cal{V}}^{\xi + \frac{\mu+1}{\tau -1}} \Delta_1 \geq  \\
& {\cal{V}}^{2\xi}-2{\cal{V}}^{\xi} \Delta_1  C_6 {\cal{V}}^{\frac43+\frac{\mu}{3(\tau-1)}}  - C \, {\cal{V}}^{\xi + \frac{\mu+1}{\tau -1}} \Delta_1      >  \\
 &\frac12 {\cal{V}}^{2\xi}
\end{split}
\label{app2}
\end{equation}
and, analogously to the previous case,  the contradiction proves the thesis.

\medskip

The same argument works also in an interval $[t, t+\Delta_\ell]$, $\ell > 1$,
assuming for the electric field
the estimate (\ref{avern}) at level $\ell - 1$.  In fact we have the bound
(see (\ref{vareps}))
$$
\langle E \rangle_{\Delta_{\ell-1}} \, \Delta_\ell  \leq C  {\cal{V}}^{-\eta-\varepsilon}
$$
  and by (\ref{nablaW'})
$$
|\nabla U(X(t))| \, \Delta_\ell \leq C {\cal{V}}^{-\eta-c}
$$
 which used in (\ref{app1}) and (\ref{app2})  allows to achieve the proof.

\bigskip

\noindent  {\bf{Proof of Lemma \ref{lemrect}}}.

\medskip

We treat first the case $\ell=1$, that is $\Delta_\ell=\Delta_1$.

\noindent Let $t_0\in [t, t+\Delta_1]\subset [0,T]$ be the time at which $|X_3(s)-Y_3(s)|$ has the minimum value. We put $\Gamma(s)=X_3(s)-Y_3(s)$.  Moreover we define the function
$$
\bar{\Gamma}(s)=\Gamma(t_0)+ \dot{\Gamma}(t_0)(s-t_0).
$$
Since the magnetic force and the external field $U$ do not act on the third component of the velocity it is:
$$
\ddot{\Gamma}(s)-\ddot{\bar{\Gamma}}(s)=E_3(X(s),s)-E_3(Y(s),s)
$$
$$
\Gamma(t_0)=\bar{\Gamma}(t_0), \quad  \dot{\Gamma}(t_0)=\dot{\bar{\Gamma}}(t_0)
$$
from which it follows
$$
\Gamma(s)=\bar{\Gamma}(s)+\int_{t_0}^s d\tau \int_{t_0}^\tau d\xi \ \big[ E_3(X(\xi),\xi)-E_3(Y(\xi),\xi) \big].
$$
By  (\ref{C1})
\begin{equation}
\begin{split}
\int_{t_0}^s d\tau\int_{t_0}^\tau d\xi \,& |E_3(X(\xi),\xi)-E_3(Y(\xi),\xi)|\leq 2C_6 {\cal{V}}^{\frac43+\frac{\mu}{3(\tau-1)}} 
\frac{|s-t_0|^2}{2}\leq 
 \\
&C_6 {\cal{V}}^{\frac43+\frac{\mu}{3(\tau-1)}}  \Delta_1  |s-t_0|\leq\frac{ |s-t_0|}{4}. 
\end{split}
\label{eq_app}
\end{equation}
 Hence,
\begin{equation}
 |\Gamma(s)|\geq |\bar{\Gamma}(s)|-\frac{|s-t_0|}{4}.
\label{z}
\end{equation}
Now we have:
$$
|\bar{\Gamma}(s)|^2=|\Gamma(t_0)|^2+2\Gamma(t_0)\dot{\Gamma}(t_0)(s-t_0)+|\dot{\Gamma}(t_0)|^2 |s-t_0|^2.
$$
We observe that $\Gamma(t_0) \dot{\Gamma}(t_0) (s-t_0)\geq 0.$  Indeed, if $t_0 \in(t, t+\Delta_1)$ 
then $\dot{\Gamma}(t_0)=0$ while if $t_0=t$ or $t_0=t+\Delta_1$ the product $\Gamma(t_0) \dot{\Gamma}(t_0) (s-t_0)\geq 0$.
Hence
$$
|\bar{\Gamma}(s)|^2\geq |\dot{\Gamma}(t_0)|^2 |s-t_0|^2.
$$
By Lemma \ref{lemv3} (adapted to this context with a factor $h\geq 1$), since $t_0\in [t, t+\Delta_1]$ it is
$$
|\dot{\Gamma}(t_0)|\geq h \frac{{\cal{V}}^{-\eta}}{2}
$$
hence
$$
|\bar{\Gamma}(s)|\geq h \frac{{\cal{V}}^{-\eta}}{2} |s-t_0|
$$
and finally by (\ref{z}), 
$$
 |\Gamma(s)|\geq h \frac{{\cal{V}}^{-\eta}}{4}|s-t_0|.
 $$
From this the thesis  follows, since obviously
$
|X(s)-Y(s)|\geq |\Gamma(s)|.$

\medskip

By the same argument used at the end of the proof of Lemma \ref{lemv3}, we see that the same proof works also 
considering the interval $[t, t+\Delta_\ell]$,  $\ell > 1$
and assuming  for the electric field
the estimate (\ref{avern}) at level $\ell - 1$.

\bigskip

\bigskip

\bigskip

\noindent$\mathbf{Acknowledgments.}$  Work performed under the auspices of GNFM-INDAM and the Italian Ministry of the University (MIUR).

\noindent One of the authors acknowledges the MIUR Excellence Department Project awarded to the Department of Mathematics, University of Rome Tor Vergata CUP E83C18000100006.

\bigskip

\bigskip

\bigskip

\section*{Appendix}
\appendix

\setcounter{equation}{0}

\def\theequation{A.\arabic{equation}}

We give here a simple physical model which motivates a potential $U$ of the form (\ref{campo_W}) for $\mu=1$.
We can think for example to put some   point charges (with charge $-1$) in the plane $x_1=0$, kept fixed at some points.
We only sketch the main features of such resulting potential, analyzing
the one originated by a single negative charge put at the origin. Considering the coulomb potential produced
by such charge, equation (\ref{compo2}), for the second component of the equation of motion
of a plasma particle, will be modified into
\begin{equation}
\begin{split}
\dot V_2^N(t) =& -  V_1^N(t) \,  h\left(X_1^N(t)\right) + E_2^N(X^N(t),t) - \frac{X_2^N(t)}{|X^N(t)|^3} \\
=& -\frac{d}{dt} {\mathcal{H}}\left(X_1^N(t)\right)
+ E_2^N(X^N(t),t) - \frac{X_2^N(t)}{|X^N(t)|^3}
\end{split}
\end{equation}
which integrated in time, 
\begin{equation}
\begin{split}
V_2^N(t)-V_2^N(0) =& \, -{\mathcal{H}}\left(X_1^N(t)\right)
+ {\mathcal{H}}\left(X_1^N(0)\right)  \\
&+ \int_0^t  E_2^N(X^N(s),s) \, ds  - \int_0^t \frac{X_2^N(s)}{|X^N(s)|^3} \, ds.
\end{split}
\label{shield}
\end{equation}
Analogously to what done after (\ref{compo2})
(recalling Proposition \ref{field}),
\begin{equation}
\begin{split}
|{\mathcal{H}}\left(X_1^N(t)\right)| \leq&  \,C {\mathcal{V}}^N(t) + \int_0^t \frac{|X_2^N(s)|}{|X^N(s)|^3} \, ds \\
\leq& \,C {\mathcal{V}}^N(t) + \int_0^t \frac{1}{|X^N(s)|^2} \, ds \\
\leq& \, C {\mathcal{V}}^N(t) + \int_0^t \frac{1}{|X_1^N(s)|^2} \, ds 
\label{ef}
\end{split}
\end{equation} 
from which, using the form of the magnetic field which holds near the border (\ref{campo_B2}),
\begin{equation}
\frac{1}{\left(X_1^N(t)\right)^{\tau-1}} \leq C {\mathcal{V}}^N(t) + C \int_0^t \left(\frac{1}{\left(X_1^N(s)\right)^{\tau-1}}\right)^{\frac{2}{\tau-1}} \, ds.
\end{equation}
From this, since $\frac{2}{\tau-1}<1$, we get
\begin{equation}
\frac{1}{\left(X_1^N(t)\right)^{\tau-1}} \leq C {\mathcal{V}}^N(t)
\end{equation}
which is the same estimate (\ref{min_dist}) obtained before, which we used to deduce the confinement of the plasma
in the region $x_1>0$.

\end{document}